\documentclass{comjnl}

\usepackage{graphicx}
\usepackage{cite}
\usepackage{amsmath}
\usepackage{amssymb}
\usepackage{array}
\usepackage{color, xcolor}
\usepackage{cite}
\usepackage{multirow}
\usepackage[flushleft]{threeparttable}
\usepackage{epstopdf}
\usepackage[T1]{fontenc}

\usepackage{aecompl}

\newcommand\blfootnote[1]{%
	\begingroup
	\renewcommand\thefootnote{}\footnote{#1}%
	\addtocounter{footnote}{-1}%
	\endgroup
}

\begin{document}
\title{Some Combinatorial Problems in Power-law Graphs}
\author{Che~Jiang\thanks{}} 
\author{Wanyue~Xu} 
\author{Xiaotian~Zhou} 
\author{Zhongzhi~Zhang} \email{zhangzz@fudan.edu.cn}
\author{Haibin~Kan} 
\affiliation{Shanghai Key Laboratory of Intelligent Information
	Processing, School of Computer Science, Fudan University, Shanghai 200433, China; \\
Shanghai Engineering Research Institute of Blockchain,  Fudan University, Shanghai, 200433, China}
\shortauthors{C. Jiang, W. Xu, X. Zhou, Z. Zhang and H. Kan}


\begin{abstract}
The power-law behavior is ubiquitous in a majority of real-world networks, and it was shown to have a strong effect on various combinatorial, structural, and dynamical properties of graphs. For example, it has been shown that in real-life power-law networks, both the matching number and the domination number are relatively smaller, compared with homogeneous graphs. In this paper, we study analytically several combinatorial problems for two power-law graphs with the same number of vertices, edges, and the same power exponent. For both graphs, we determine exactly or recursively their matching number, independence number, domination number, the number of maximum matchings, the number of maximum independent sets, and the number of minimum dominating sets. We show that power-law behavior itself cannot characterize the combinatorial properties of a heterogenous graph. Since the combinatorial properties studied here have found wide applications in different fields, such as structural controllability of complex networks, our work offers insight in the applications of these combinatorial problems in power-law graphs.
\end{abstract}

\maketitle

\section{Introduction}\label{sec:introduction}

Let  $\mathcal{G}=(\mathcal{V} ,\mathcal{E})$ be a connected unweighted graph with vertex set $\mathcal{V}$  and edge set $\mathcal{E}$. A matching of  graph $\mathcal{G}$ is a subset of edge set $\mathcal{E}$, where no two edges are incident to a common vertex. A  matching of maximum cardinality is called a maximum matching. The matching number of graph $\mathcal{G}$ is the cardinality of a maximum matching. An independent set of a graph $\mathcal{G}$ is a subset $\mathcal{I}$ of vertex set $\mathcal{V}$, such that each pair of vertices in  $\mathcal{I}$ is not adjacent in $\mathcal{G}$. A maximum independent set (MIS) is  an independent set  $\mathcal{I}$ with the largest cardinality.  The cardinality of a MIS for graph $\mathcal{G}$ is called its independent number. Graph $\mathcal{G}$ is called a unique independence graph if it has a unique MIS~\cite{HoSt85}.  A dominating set of a graph $\mathcal{G}$ is a subset $\mathcal{D}$ of vertex set $\mathcal{V}$, such that every vertex in $\mathcal{V} \setminus \mathcal{D}$ is connected to at least one vertex in set $\mathcal{D}$. A dominating set $\mathcal{D}$ is called a minimum dominating set (MDS) if it has the least cardinality. The cardinality of a MDS for graph $\mathcal{G}$ is called its domination number.

\blfootnote{*Currently at: Department of Physics, Fudan University, Shanghai 200433, China}

The aforementioned combinatorial problems have been applied to numerous aspects in various disciplines or practical areas. For example, the size and the number of maximum matchings have found applications in physics~\cite{Mo64}, chemistry~\cite{Vu11}, computer science~\cite{LoPl86}; the MIS problem is associated with many fundamental graph problems, being equivalent to the minimum vertex cover problem~\cite{Ka72} in the same graph and the maximum clique problem in its complement graph~\cite{PaXu94}, and has been widely used to collusion detection in voting pools~\cite{ArFaDoSiKo11} and wireless networking schedules~\cite{JoLiRySh16}; while the MDS problem is closely related to multi-document summarization in sentence graphs~\cite{ShLi10}, routing on ad hoc wireless networks~\cite{Wu02}, and controllability in protein interaction networks~\cite{Wu14}. Of particular interest is the connection of maximum matchings and MDS to structural controllability of complex networks~\cite{LiBa16}, in the contexts of  vertex~\cite{LiSlBa11} and edge~\cite{NeVi12} dynamics, respectively.

In view of the intrinsic relevance in both theoretical and practical scenarios, the above combinatorial problems have received considerable attention from the scientific community of theoretical computer science, theoretical physics, discrete mathematics, among others. In the past decade, these problems have become very active and have been popular research objects. Many authors have devoted their efforts to developing algorithms for the problems associated with maximum matchings~\cite{YaZh05,YaZh08, ChFrMe10,Yu13, ZhWu15, LiZh17}, MISs~\cite{XiNa13,HoKlLiLiPoWa15,ChEn16}, as well as MDSs~\cite{FoGrPySt08, HeIs12,NaAk12,GaHaK15,ShLiZh17}. Although  scientists have made a concerted effort, solving these problems is an important challenge and often computationally difficult. For example, finding a MDS~\cite{HaHeSl98} or a MIS~\cite{Ro86,HaRa97} of a general graph is NP-hard; while enumerating maximum matchings, or MISs, or MDSs  in a graph is more difficult, which is \#P-complete even in a bipartite graph~\cite{Va79TCS,Va79SiamJComput}. Thus, it makes sense to construct or seek special graph classes for which these combinatorial problems can be exactly solved~\cite{LoPl86}, in order to achieve a particular goal.

On the other hand, extensive empirical study~\cite{Ne03} has uncovered that a majority of real-world networks are typically scale-free~\cite{BaAl99}, characterized by a power-law distribution $P(k) \sim k^{-\gamma}$ ($2< \gamma \leq 3$) for their vertex degree. This nontrivial scale-free structure is a fundamental concept in the study of the emerging network sciences. Many previous studies have shown that the scale-free topology plays an important role in various structural~\cite{ChLu02}, combinatorial~\cite{LiSlBa11,ZhWu15, NaAk12, GaHaK15, ShLiZh17}, and dynamical~\cite{AlJeBa00,ChWaWaLeFa08,YiZhPa19} properties of a graph. In the context of combinatorial aspect, it has been shown that compared with non-scale-free graphs, in scale-free networks, both the matching number~\cite{LiSlBa11} and the number of maximum matchings~\cite{ZhWu15} are significantly smaller. It is the same with the domination number and the number of MDSs~\cite{ NaAk12, GaHaK15, LiZh17}. In addition, scale-free architecture also strongly affects the MIS problem~\cite{ShLiZh18} and its related optimization algorithms~\cite{FePaPa08}. As is well known, in addition to the scale-free topology, many real networks show simultaneously some other remarkable properties, e.g., self-similarity~\cite{SoHaMa05}. Thus, it is difficult to separate the role or effect of a specific structural property in the performance of a network. Then, an interesting question is raised naturally: whether the scale-free structure is a unique ingredient characterizing the above combinatorial problems in power-law graphs?

In this paper, we study several combinatorial problems in two self-similar scale-free networks with the same power exponent: one is fractal but not small-world~\cite{ZhLiWuZo11}, the other is small-world but not fractal~\cite{HiBe06}. For both graphs, by using the decimation technique based on their self-similarity, we determine exactly the matching number, the independence number, and the domination number. Moreover, we determine exactly or recursively the number of maximum matchings, the number of MISs, and the number of MDSs. We show that the two networks differ in the studied quantities, which implies that scale-free topology alone cannot determine the combinatorial properties of power-law graphs, including maximum matchings, MISs and MDSs. Moreover, our exact results are instrumental for testing heuristic or stochastic algorithms associated with related combinatorial problems.

\section{Constructions and structural properties of self-similar scale-free networks}

In this section, we give a brief introduction to constructions and their structural properties of two self-similar scale-free networks, with the same number of vertices,  the same number of edges, and the same power exponent.   One is fractal but not small-world~\cite{ZhLiWuZo11}, the other is small-world but not fractal~\cite{HiBe06}.

\subsection{Constructions and structural properties of fractal scale-free networks}

We first introduce the fractal scale-free networks under consideration,  which are generated by an iterative way.
Let $\mathcal{G}_n$, $n \geq 0$, denote the fractal scale-free network after $n$ iterations. Then, $\mathcal{G}_n$ is constructed as follows: For $n=0$, $\mathcal{G}_0$ is the complete graph $\mathcal{K}_2$ with two vertices connected by an iterative edge. For $n\geq 1$, $\mathcal{G}_n$ is obtained from $\mathcal{G}_{n-1}$  by performing the following operation: replace each iterative edge by the connected cluster   on the right-hand side (rhs) of the arrow  in Fig.~\ref{fcons0}.

\begin{figure}
\centering
\includegraphics[width=0.3\textwidth]{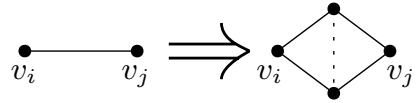}
\caption{Construction method for the fractal scale-free networks. To obtain network of next iteration, each iterative edge $(v_i,v_j)$ of current iteration is replaced by two parallel paths of two iterative edges (solid lines) on the rhs of the arrow, with $v_i$ and $v_j$ being the end vertices of two paths, and then link the two new vertices other than $v_i$ and $v_j$ by a new non-iterative edge (dotted line). }
\label{fcons0}
\end{figure}

Figure~\ref{fcons} illustrates the construction process of the first several iterations.

\begin{figure}
\centering
\includegraphics[width=0.45\textwidth]{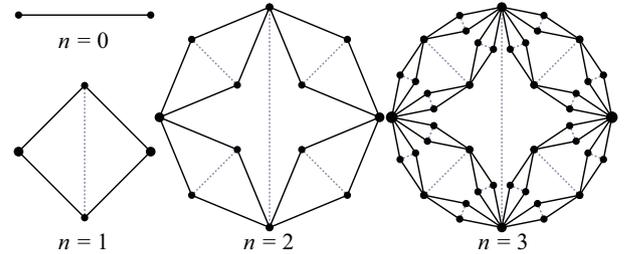}
\caption{The first three iterations of the fractal scale-free networks.}
\label{fcons}
\end{figure}

The fractal scale-free networks are self-similar, which can be easily seen from an alternative construction approach~\cite{ZhLiWuZo11} as shown in Fig.~\ref{fcons2}.  For $\mathcal{G}_n$, $n\geq0$, we call the two vertices in $\mathcal{G}_0$ as initial vertices, and denote them as $X_n$ and $Y_n$; while call the two vertices generated at iteration $1$ as hub vertices, and denote them as $W_n$ and $Z_n$. Then, given  the network $\mathcal{G}_{n}$, $n\geq 1$, $\mathcal{G}_{n+1}$ can be obtained by  merging four copies of $\mathcal{G}_{n}$ at their initial vertices. Let $\mathcal{G}_{n}^{(\theta)}$, $\theta=1,2,3,4$, be four replicas of $\mathcal{G}_{n}$, and denote  the two  initial vertices of $\mathcal{G}_{n}^{(\theta)}$  by $X_{n}^{(\theta)}$ and $Y_{n}^{(\theta)}$, respectively. Then, $\mathcal{G}_{n+1}$ can be obtained by merging $\mathcal{G}_{n}^{(\theta)}$, $\theta=1,2,3,4$,  with $X_{n}^{(1)}$
($Y_{n}^{(2)}$) and $X_{n}^{(3)}$ ($Y_{n}^{(4)}$) being identified as the initial vertex $X_{n+1}$ ($Y_{n+1}$) in $\mathcal{G}_{n+1}$, while $Y_{n}^{(1)}$
($Y_{n}^{(3)}$) and $X_{n}^{(2)}$ ($X_{n}^{(4)}$) being identified as the hub vertex $W_{n+1}$ ($Z_{n+1}$) in $\mathcal{G}_{n+1}$.  After the joining process, we link the two hub vertices $W_{n+1}$ and $Z_{n+1}$ by a non-iterative edge and get $\mathcal{G}_{n+1}$.

\begin{figure}
\centering
\includegraphics[width=0.45\textwidth]{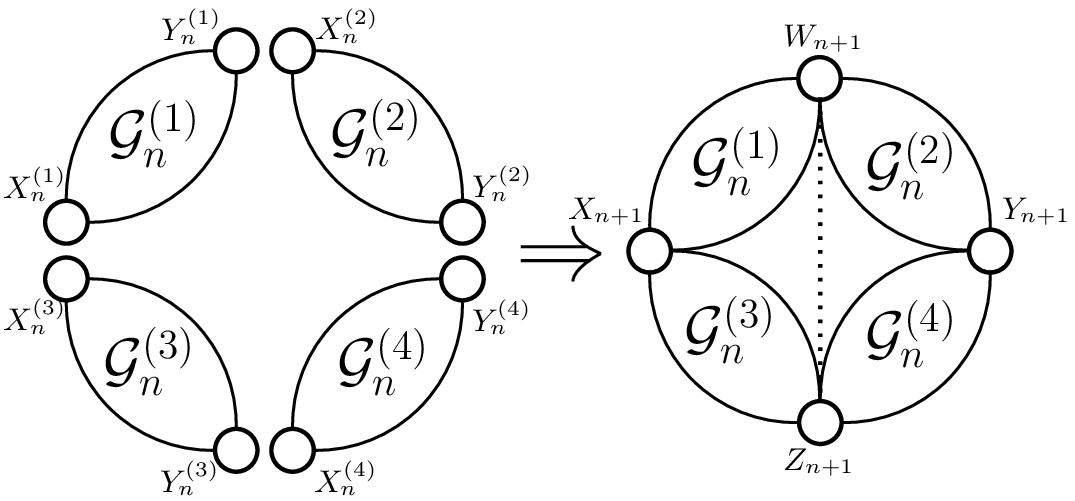}
\caption{Second approach for the construction of $\mathcal{G}_{n+1}$.}
\label{fcons2}
\end{figure}

Let $N_n$ and $E_n$, respectively, stand for the number of vertices and  the number of edges in $\mathcal{G}_n$. By the second construction rules, $N_n$ and $E_n$ satisfy relations $N_n = 4N_{n-1}-4$ and $E_n = 4E_{n-1}+1$. With the initial conditions $N_1 = 4$ and $E_1 = 5$, we have $N_n = \frac{2}{3}\left(4^{n}+2 \right)$ and $E_n=\frac{1}{3}\left(4^{n+1}-1 \right)$. Then the average degree of all vertices in $\mathcal{G}_n$  is $\frac{2  E_n}{N_n}=\frac{2(4^{n+1}-1)}{2 \times 4^n + 4}$, which is asymptotically equal to  $4$ for large $n$.

The resulting graph  $\mathcal{G}_n$ is  scale-free, since the degree of its vertices obeys a power-law distribution $P(k) \propto k^{-3}$. Moreover, it is fractal with a fractal dimension being $2$~\cite{ZhLiWuZo11}. However, it is not small-world, since for  large $n$, the average distance $\bar{d}_n$ of  $\mathcal{G}_n$ grows as a power function of  $N_n $, that is, $\bar{d}_n \sim (N_n)^{1/2}$.

\subsection{Constructions and structural properties of non-fractal scale-free networks}

The second networks we consider are non-fractal and scale-free, which are also constructed iteratively.  Let $\mathcal{G}'_n$, $n \geq 0$, denote the network after $n$ iterations. Then, $\mathcal{G}'_n$ is  built as follows. For $n=0$, $\mathcal{G}'_0$ is the complete graph $\mathcal{K}_2$ with two vertices connected by an iterative edge. For $n\geq 1$, $\mathcal{G}'_n$ is obtained from $\mathcal{G}'_{n-1}$  by performing the following operation: replace each iterative edge by the connected cluster   on the rhs of the arrow  in Fig.~\ref{nfcons0}.

\begin{figure}
\centering
\includegraphics[width=0.3\textwidth]{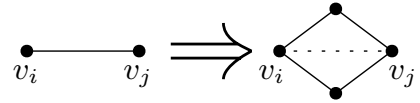}
\caption{Construction method for the non-fractal scale-free networks. To obtain network of next iteration, each iterative edge $(v_i,v_j)$ of current iteration is replaced by two parallel paths of two iterative edges (solid lines) on the rhs of the arrow, with $v_i$ and $v_j$ being the end vertices of two paths, and then link $v_i$ and $v_j$ by a new non-iterative edge (dotted line).}
\label{nfcons0}
\end{figure}

Figure~\ref{nfcons} illustrates the construction process of the first several iterations.

\begin{figure}
\centering
\includegraphics[width=0.45\textwidth]{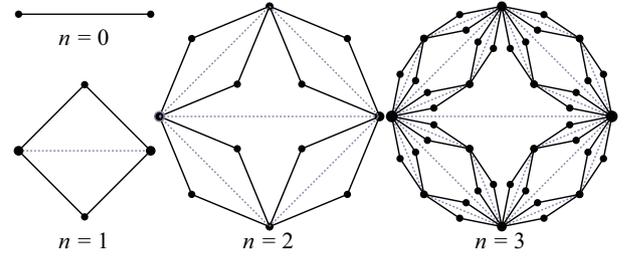}
\caption{The first three iterations of the non-fractal scale-free networks.}
\label{nfcons}
\end{figure}

The non-fractal scale-free network $\mathcal{G}'_n$ is also self-similar, which suggests another construction approach highlighting its self-similarity~\cite{HiBe06} as shown in Fig.~\ref{nfcons2}.  For $\mathcal{G}'_n$, $n\geq0$, we call the two vertices in $\mathcal{G}'_0$ as hub vertices, and denote them as $X'_n$ and $Y'_n$; while call the two vertices generated at iteration $1$ as border vertices, and denote them as $W'_n$ and $Z'_n$. Then, given  the network $\mathcal{G}'_{n}$, $n\geq 1$, $\mathcal{G}'_{n+1}$ can be obtained by  merging four copies of $\mathcal{G}'_{n}$ at their hub vertices. Let $\mathcal{G}_{n}^{'(\theta)}$, $\theta=1,2,3,4$, be four replicas of $\mathcal{G}'_{n}$, and denote  the two hub vertices of $\mathcal{G}_{n}^{'(\theta)}$  by $X_{n}^{'(\theta)}$ and $Y_{n}^{'(\theta)}$, respectively. Then, $\mathcal{G}'_{n+1}$ can be obtained by merging $\mathcal{G}_{n}^{'(\theta)}$, $\theta=1,2,3,4$,  with $X_{n}^{'(1)}$
($Y_{n}^{'(2)}$) and $X_{n}^{'(3)}$ ($Y_{n}^{'(4)}$) being identified as the hub vertex $X'_{n+1}$ ($Y'_{n+1}$) in $\mathcal{G}'_{n+1}$, while $Y_{n}^{'(1)}$
($Y_{n}^{'(3)}$) and $X_{n}^{'(2)}$ ($X_{n}^{'(4)}$) being identified as the bounder  vertex $W'_{n+1}$ ($Z'_{n+1}$) in $\mathcal{G}'_{n+1}$.  After the joining process, we link the two hub vertices $X'_{n+1}$ and $Y'_{n+1}$ by a non-iterative edge and get $\mathcal{G}'_{n+1}$.

\begin{figure}
\centering
\includegraphics[width=0.45\textwidth]{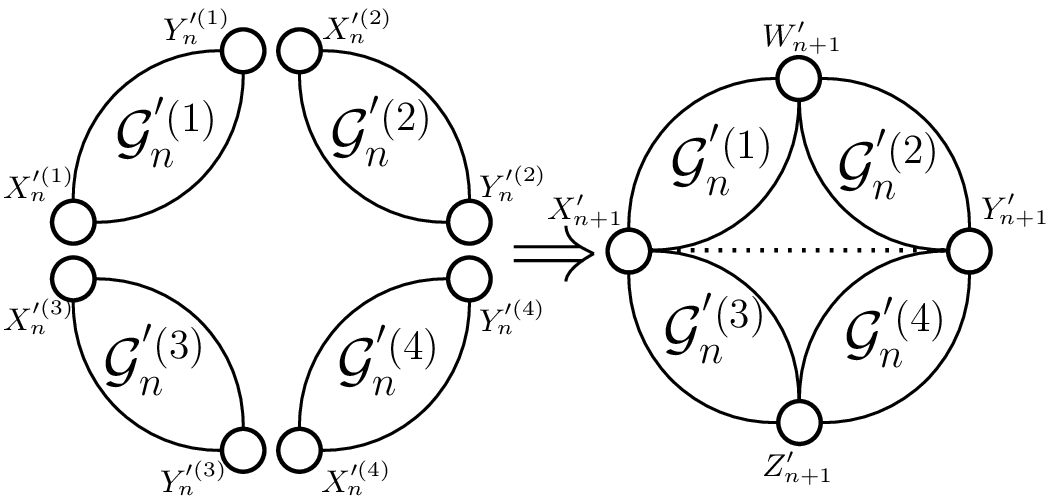}
\caption{Second  approach for the construction of $\mathcal{G}'_{n+1}$.}
\label{nfcons2}
\end{figure}

By construction, $\mathcal{G}'_n$ has the same number of vertices $N_n$, the same number of edge $E_n$, and thus the same average degree as those of  $\mathcal{G}_n$.  Moreover, the $\mathcal{G}'_n$ is also scale-free with the same power exponent 3 as that of $\mathcal{G}_n$. However, different from $\mathcal{G}_n$,  $\mathcal{G}'_n$ is non-fractal since its fractal dimension is infinite, but is small-world, with its average distance  average distance $\bar{d}_n$  growing logarithmically with the number of vertices  $N_n$.

After introducing the construction and topological properties of the two self-similar scale-free networks,  in what follows, by using their self-similarity we will study some combinatorial problems for these two networks, including the matching number, the independence number, the domination number, the number of maximum matchings, the number of MISs, and the number of MDSs. We will show that for the studied quantities, the two networks exhibit  quite different behaviors. We note that in the process of the following computation or proof, we employ the same notation  for $\mathcal{G}_n$ and $\mathcal{G}'_n$ in the case without inducing  confusion.

\section{Matching number and the number of maximum matchings}

In this section, we study the matching number and the number of maximum matchings in the self-similar scale-free networks.

\subsection{Matching number and the number of maximum matchings in fractal scale-free networks}

We first study the matching number and the number of maximum matchings in graph $\mathcal{G}_n$.

\subsubsection{Matching number}

Let $\beta_n$ denote the matching number of graph $\mathcal{G}_n$. In order to determine $\beta_n$, we define some intermediate quantities. Note that according to the number of covered initial vertices,  all the matchings of $\mathcal{G}_n$ can be classified into three types: $\Omega^0_n$, $\Omega^1_n$ and $\Omega^2_n$, where $\Omega^k_n$, $k=0,1,2$, represent the set of matchings with each  covering exactly $k$ initial vertices  of $\mathcal{G}_n$.  Let $\Theta^k_n$, $k=0,1,2$, be the subset of $\Omega^k_n$, where each matching has the largest  cardinality, denoted by $\beta^k_n$, $k=0,1,2$. Then,  $\beta_n={\rm max}\{\beta^0_n, \beta^1_n,\beta^2_n\}$.

\begin{theorem}
The matching number of graph $\mathcal{G}_n$ is $\beta_n=\frac{4^{n}+2}{3}$.
\end{theorem}
\begin{proof}
Since $\beta_n={\rm max}\{\beta^0_n, \beta^1_n,\beta^2_n\}$,  we next evaluate the three quantities $\beta^0_n$, $\beta^1_n$ and $\beta^2_n$, all of which can be determined graphically.

Figures~\ref{mfo0},~\ref{mfo1}, and~\ref{mfo2} show, respectively, all the available configurations of maximum matchings of graph $\mathcal{G}_{n+1}$ belonging to $\Omega^k_{n+1}$,  $k=0,1,2$, which contains all the matchings in $\Theta^k_{n+1}$. In Figs.~\ref{mfo0},~\ref{mfo1}, and~\ref{mfo2}, only the initial vertices $X_{n}^{(\theta)}$ and $Y_{n}^{(\theta)}$ of $\mathcal{G}_{n}^{(\theta)}$, $\theta=1,2,3,4$, forming $\mathcal{G}_{n+1}$ are shown explicitly, with filled circles representing covered vertices and empty circles representing vacant vertices. Note that in Figs.~\ref{mfo0},~\ref{mfo1}, and~\ref{mfo2}, if both of the hub vertices, $W_{n+1}$ and  $Z_{n+1}$, of $\mathcal{G}_{n+1}$ are vacant, then the non-iterative edge connecting them is  included in the matching in order to maximize its cardinality. From these three figures, we establish the following recursion relations for $\beta^0_n$, $\beta^1_n$, and $\beta^2_n$:
\begin{align} 
\beta^0_{n+1}&={\rm max}\{4\beta^0_n+1 ,3\beta^0_n+\beta^1_n ,2\beta^0_n+2\beta^1_n\}\,, \label{beta0F}\\
\beta^1_{n+1}&={\rm max}\{3\beta^0_n+\beta^1_n+1 ,3\beta^0_n+\beta^2_n ,2\beta^0_n+2\beta^1_n ,  \notag\\
&\quad \quad\quad\quad
2\beta^0_n+\beta^1_n+\beta^2_n, \beta^0_n+3\beta^1_n \}\,,  \label{beta1F}\\
\beta^2_{n+1}&={\rm max}\{2\beta^0_n+2\beta^1_n+1 ,2\beta^0_n+\beta^1_n+\beta^2_n ,\beta^0_n+3\beta^1_n ,  \notag\\
&\quad \quad\quad\quad
2\beta^0_n+2\beta^2_n,\beta^0_n+2\beta^1_n+\beta^2_n ,4\beta^1_n \}\,.  \label{beta2F}
\end{align}
With initial condition $\beta^0_1=1$, $\beta^1_1=1$, and $\beta^2_1=2$, the above equations are solved to yield $\beta^0_n=\frac{4^{n}-1}{3}$, $\beta^1_n=\frac{4^{n}-1}{3}$, and $\beta^2_n=\frac{4^{n}+2}{3}$.
\end{proof}

Since the number of vertices $N_n$ in graph $\mathcal{G}_{n}$ is $N_n = \frac{2}{3}\left(4^{n}+2 \right)$, which is exactly twice as large as the matching number $\beta_n=\frac{4^{n}+2}{3}$,   there are perfect matchings in $\mathcal{G}_{n}$ for all $n\geq 0$.

\begin{figure}
    \centering
    \begin{minipage}[c]{0.5\textwidth}
        \centering
        \includegraphics[width=0.27\textwidth]{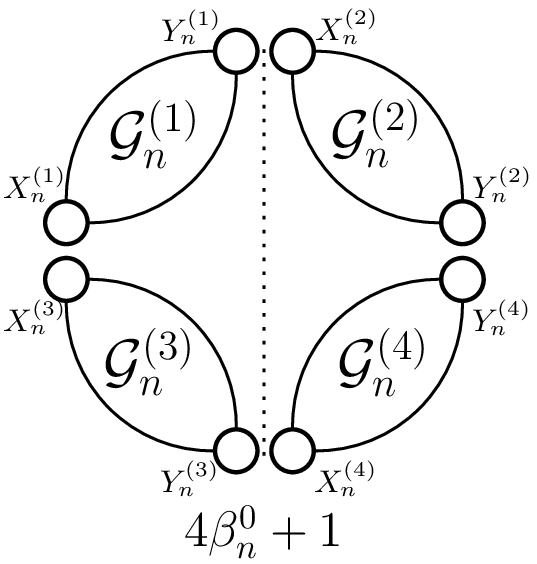}
        \includegraphics[width=0.3\textwidth]{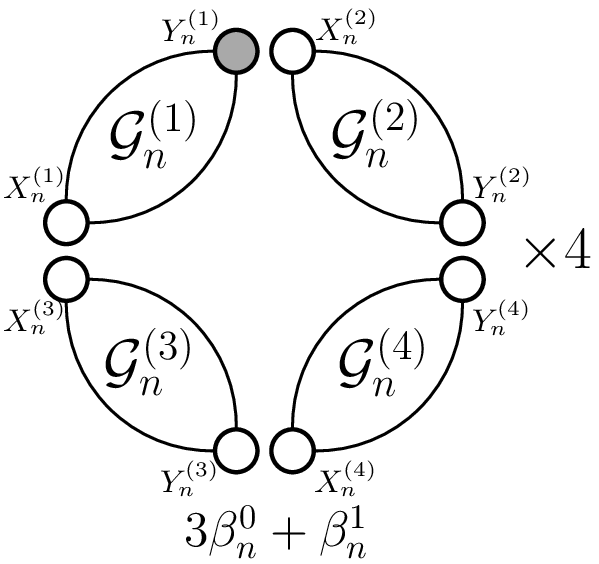}
        \includegraphics[width=0.3\textwidth]{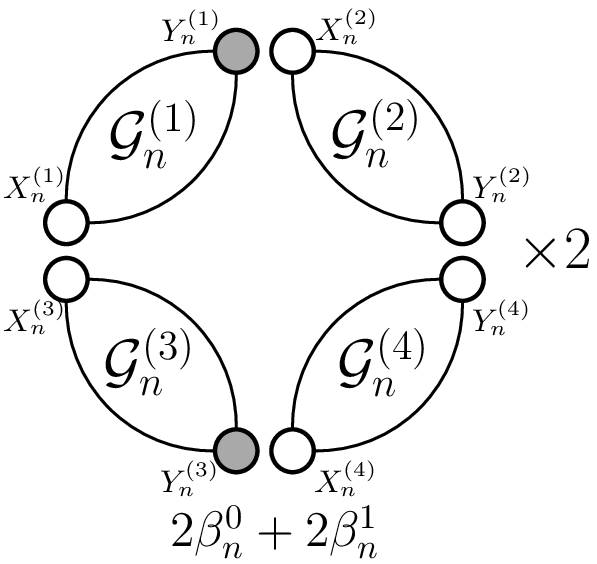}
        \\ \hspace*{\fill} \\
    \end{minipage}
    \begin{minipage}[c]{0.5\textwidth}
        \centering
        \includegraphics[width=0.3\textwidth]{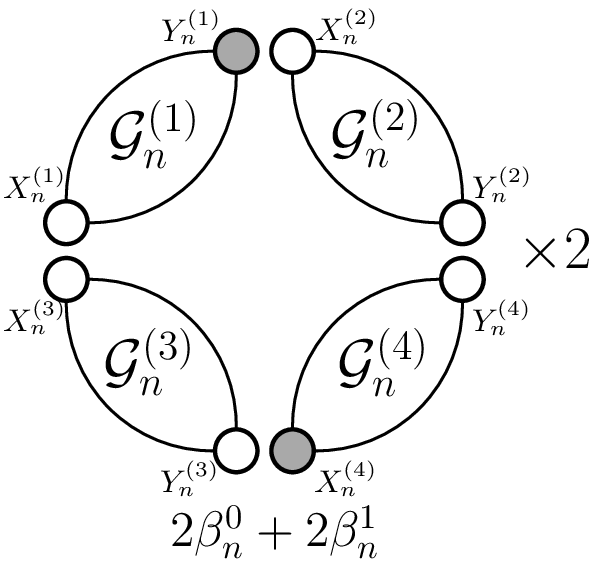}
    \end{minipage}
\caption{Illustration of all possible configurations and their sizes of matchings for graph $\mathcal{G}_{n+1}$ belonging to $\Omega^0_{n+1}$, which contain all matchings in $\Theta^0_{n+1}$. }
\label{mfo0}
\end{figure}

\begin{figure}
    \centering
    \begin{minipage}[c]{0.5\textwidth}
        \centering
        \includegraphics[width=0.3\textwidth]{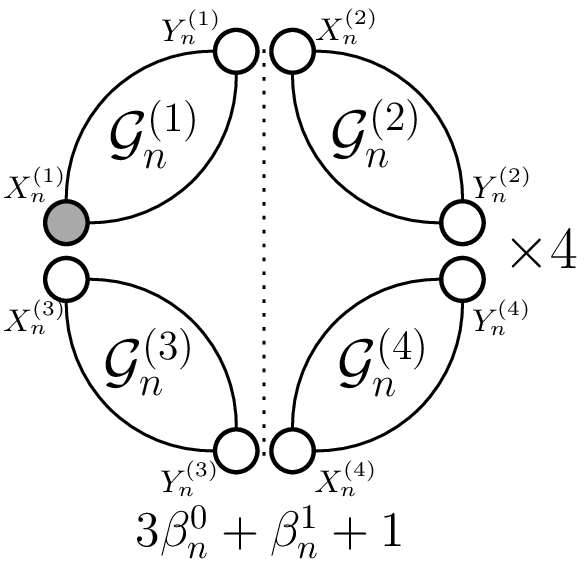}
        \includegraphics[width=0.3\textwidth]{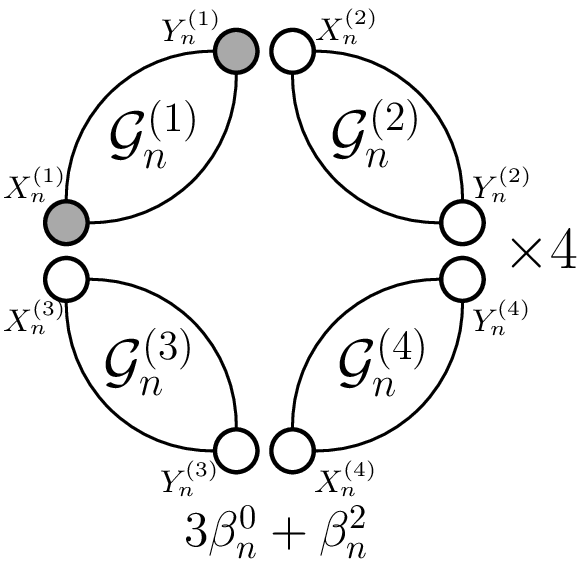}
        \includegraphics[width=0.3\textwidth]{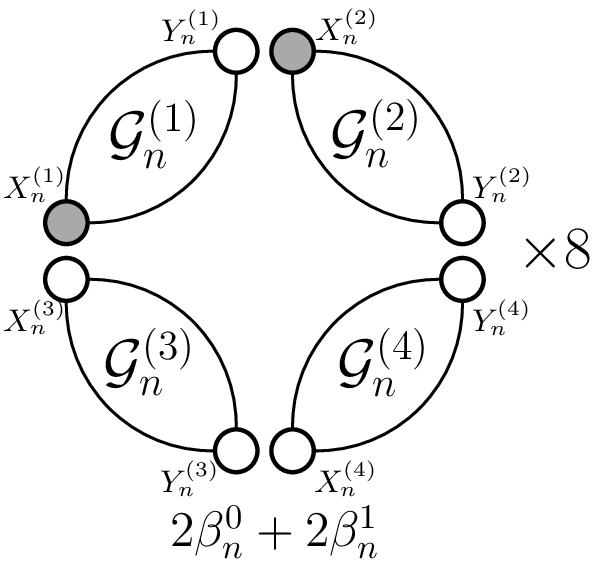}
        \\ \hspace*{\fill} \\
    \end{minipage}
    \begin{minipage}[c]{0.5\textwidth}
        \centering
        \includegraphics[width=0.3\textwidth]{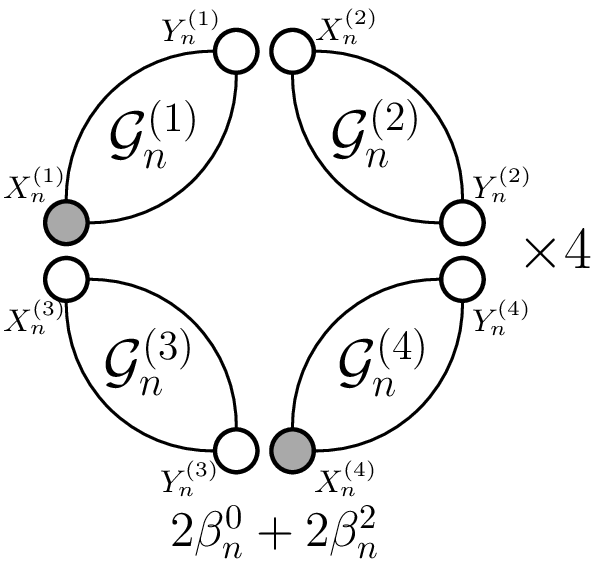}
        \includegraphics[width=0.3\textwidth]{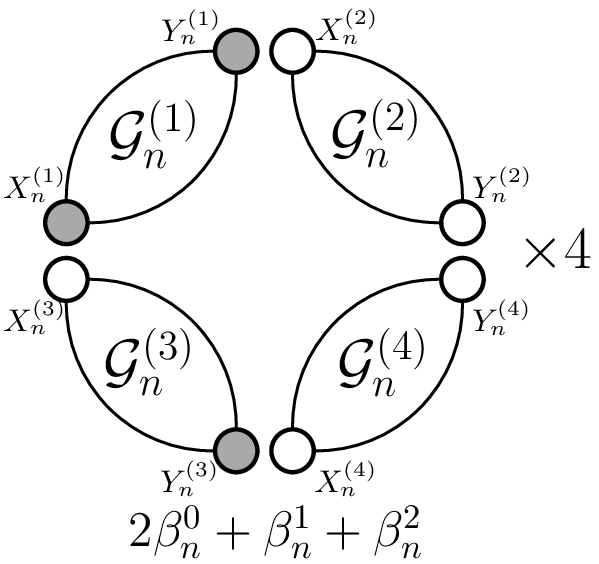}
        \includegraphics[width=0.3\textwidth]{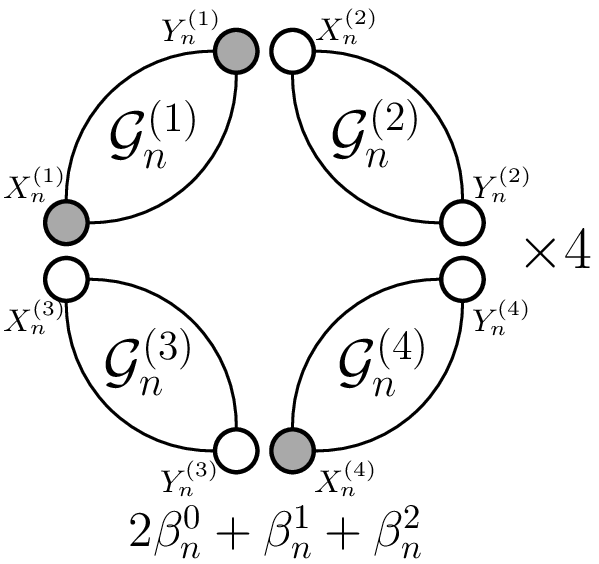}
        \\ \hspace*{\fill} \\
    \end{minipage}
    \begin{minipage}[c]{0.5\textwidth}
        \centering
        \includegraphics[width=0.3\textwidth]{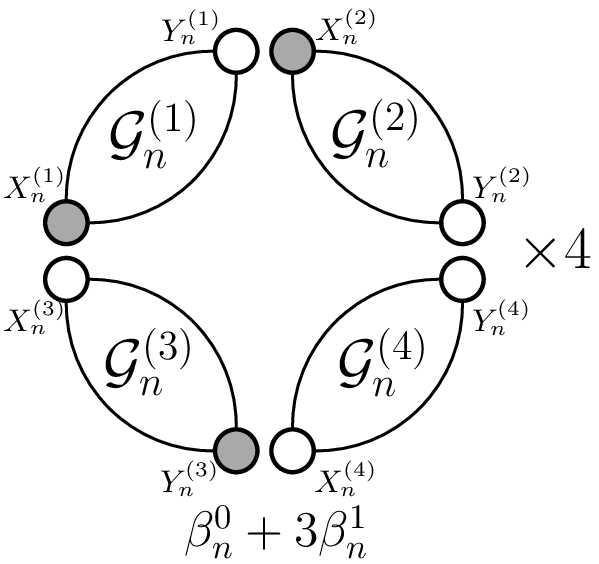}
    \end{minipage}
\caption{Illustration of all possible configurations and their sizes of matchings for graph $\mathcal{G}_{n+1}$ belonging to $\Omega^1_{n+1}$, which contain all matchings in $\Theta^1_{n+1}$. }
\label{mfo1}
\end{figure}

\begin{figure}
    \centering
    \begin{minipage}[c]{0.5\textwidth}
    \centering
        \includegraphics[width=0.3\textwidth]{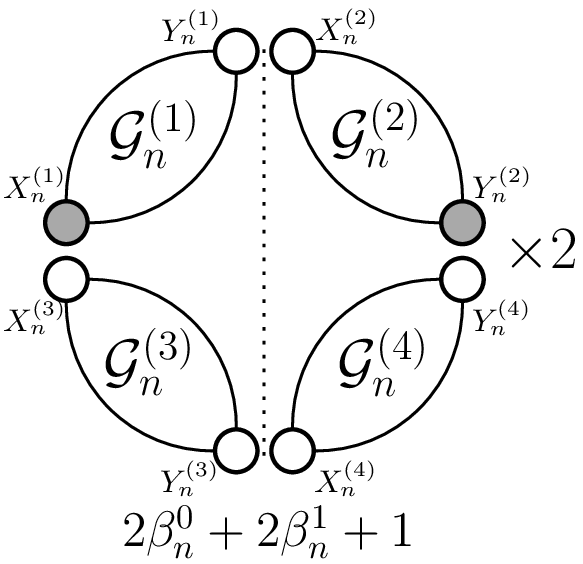}
        \includegraphics[width=0.3\textwidth]{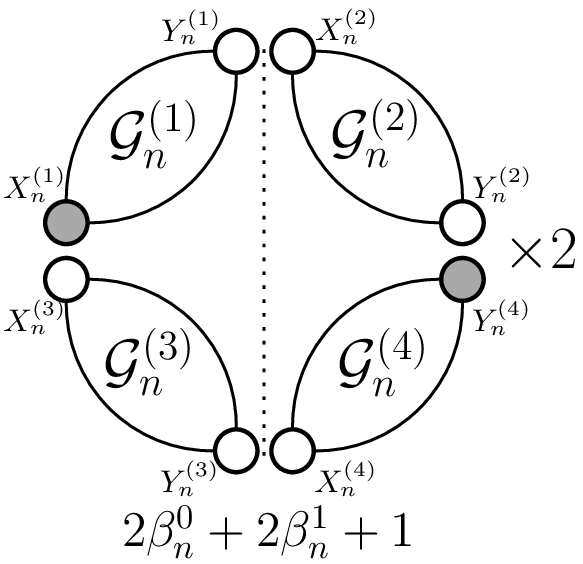}
        \includegraphics[width=0.3\textwidth]{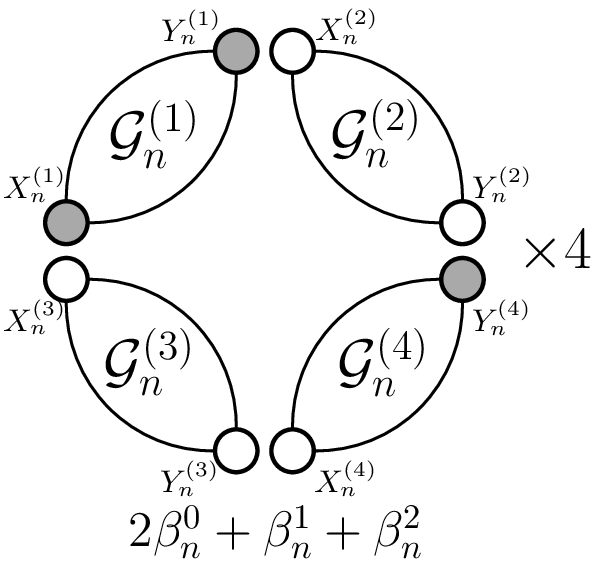}
        \\ \hspace*{\fill} \\
    \end{minipage}
    \begin{minipage}[c]{0.5\textwidth}
    \centering
        \includegraphics[width=0.3\textwidth]{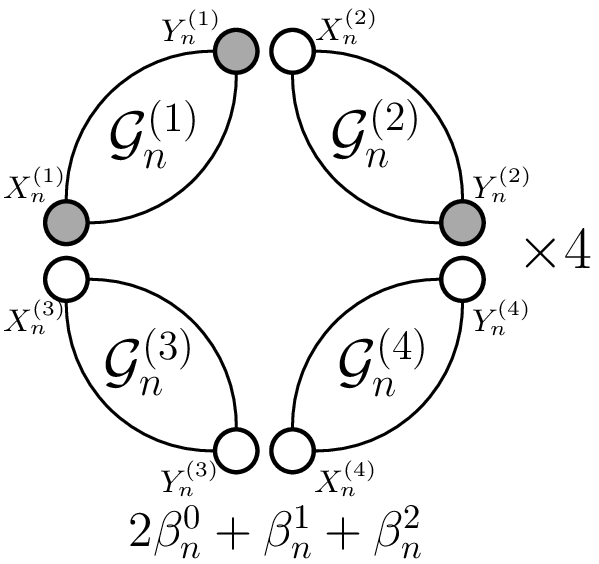}
        \includegraphics[width=0.3\textwidth]{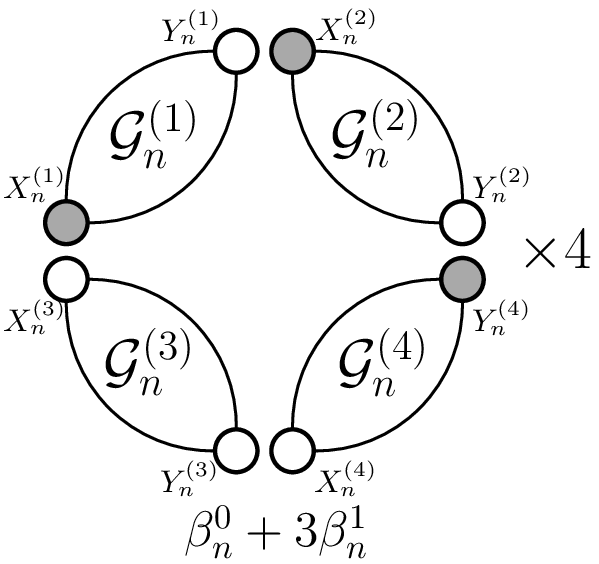}
        \includegraphics[width=0.3\textwidth]{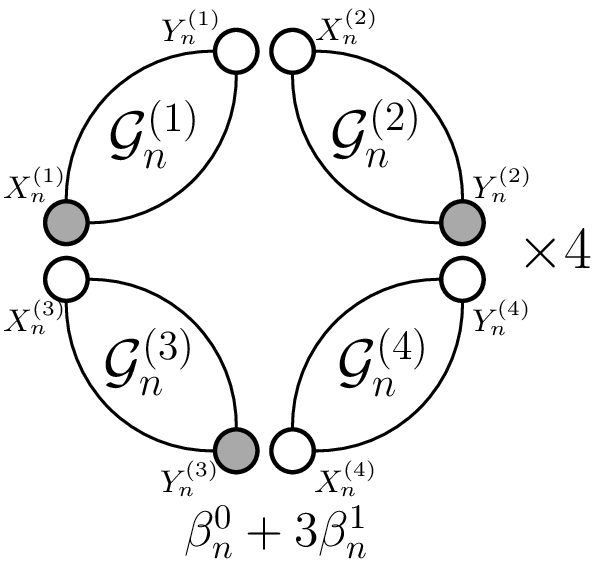}
        \\ \hspace*{\fill} \\
    \end{minipage}
    \begin{minipage}[c]{0.5\textwidth}
        \centering
        \includegraphics[width=0.3\textwidth]{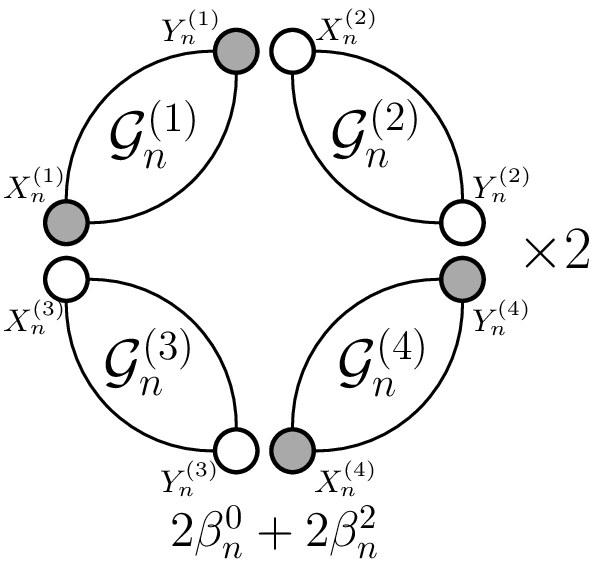}
        \includegraphics[width=0.3\textwidth]{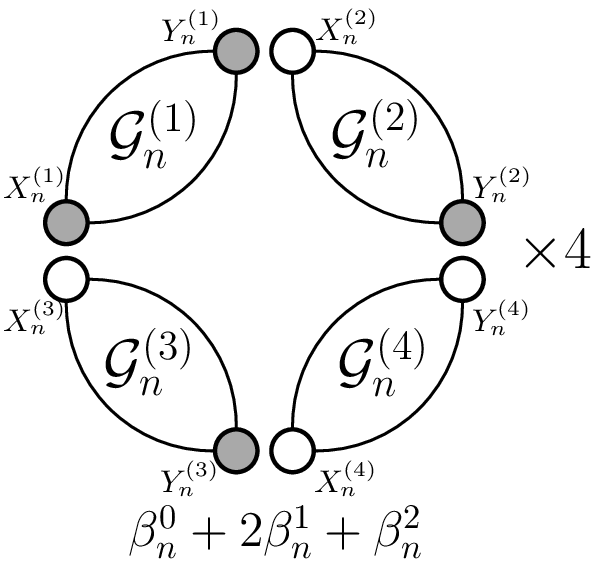}
        \includegraphics[width=0.3\textwidth]{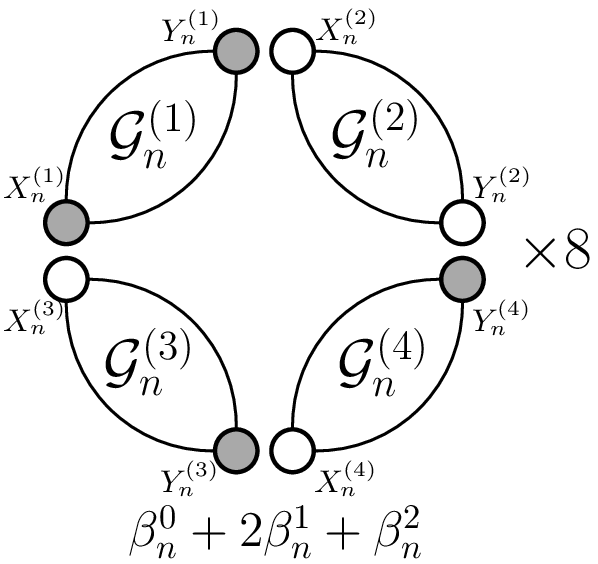}
        \\ \hspace*{\fill} \\
    \end{minipage}
    \begin{minipage}[c]{0.5\textwidth}
        \centering
        \includegraphics[width=0.3\textwidth]{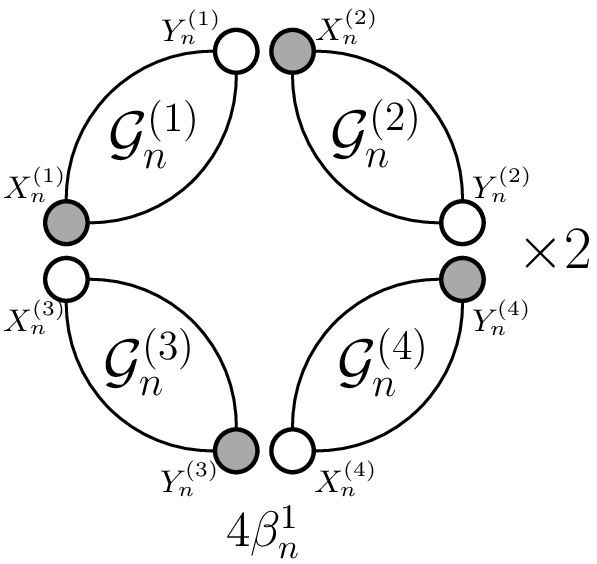}
    \end{minipage}
\caption{Illustration of all possible configurations and their sizes of matchings for graph $\mathcal{G}_{n+1}$ belonging to $\Omega^2_{n+1}$, which contain all matchings in $\Theta^2_{n+1}$. }
\label{mfo2}
\end{figure}

\subsubsection{Number of maximum matchings}

Let $\theta_n$ denote the number of maximum matchings or perfect matchings  in $\mathcal{G}_n$. To calculate $\theta_n$, we introduce an additional quantity  $\phi_n$, which denotes  the number of maximum matchings in $\Omega^0_n$, satisfying that each matching is maximum among all the matchings of $\mathcal{G}_n$ with both initial vertices $X_n$ and $Y_n$ being vacant.

\begin{theorem}\label{mnf}
The number of maximum matchings of $\mathcal{G}_n$, $n \geq 1$, is $2^{2^n-1}$.
\end{theorem}
\begin{proof}
Note that for $n=1$,  $\theta_1 = 2$ and $\phi_1 = 1$. For $n \geq 1$, we first establish the  following recursion relations for the two quantities $\theta_n$ and $\phi_n$ associated with graph $\mathcal{G}_n$:
\begin{align}
\theta_{n+1}=&2\theta_n^2\phi_n^2, \label{thetaF}\\
\phi_{n+1}=&\phi_n^4.\label{phiF}
\end{align}

We only prove Eq.~\eqref{phiF}. Since $\beta^0_n=\beta^1_n$,   Eq.~\eqref{beta0F} and Fig.~\ref{mfo0} show that, the cardinality of each matching in  $\Omega^0_{n+1}$ is maximized if and only if all the matchings of the four copies  $\mathcal{G}_n^{(\theta)}$, $\theta=1,2,3,4$,  are in $\Omega^0_n$. Then, we  establish $\phi_{n+1}=\phi_n^4$.

By using  Eq.~\eqref{beta0F}  and Fig.~\ref{mfo2}, Eq.~\eqref{thetaF}  can be proved analogously.

Equations~\eqref{thetaF} and~\eqref{phiF}, together with the initial conditions $\theta_1 = 2$ and $\phi_1 = 1$, are solved to yield $\theta_n=2^{2^n-1}$ and $\phi_{n}=1$ for all $n \geq 1$.
\end{proof}

Note that both the matching number and the number of maximum matchings for graph $\mathcal{G}_n$ have been previously obtained in~\cite{ZhWu15} by using the technique of Pfaffian orientations, which is more complicated than the approach used here.


\subsection{Matching number and the number of maximum matchings in non-fractal scale-free networks}

We continue to study the matching number and the number of maximum matchings in graph $\mathcal{G}'_n$.

\subsubsection{Matching number}

Let $\Omega^k_n$,  $k=0,1,2$, represent matchings covering exactly $k$ hub vertices  of $\mathcal{G}'_n$. Let $\Theta^k_n$, $k=0,1,2$, be the subset of $\Omega^k_n$, where each matching has the largest cardinality among all matchings in $\Omega^k_n$, with the largest cardinality being  denoted by $\beta^k_n$, $k=0,1,2$. Then,  the matching number $\beta_n$ of $\mathcal{G}'_n$ can be expressed as  $\beta_n={\rm max}\{\beta^0_n, \beta^1_n,\beta^2_n\}$.

\begin{theorem}
The matching number of graph $\mathcal{G}'_n$ is $\beta_n=\frac{2^{2n-1}+4}{3}$.
\end{theorem}
\begin{proof}
In order to find $\beta_n$, we can alternatively  evaluate the three quantities $\beta^0_n$, $\beta^1_n$ and $\beta^2_n$ by using the self-similar structure of graph $\mathcal{G}'_n$. We now graphically compute $\beta^k_n$, $k=0,1,2$.  Figures~\ref{mnfo0},~\ref{mnfo1}, and~\ref{mnfo2} show, respectively, all the possible configurations of matchings in $\Omega^0_{n+1}$, $\Omega^1_{n+1}$ and $\Omega^2_{n+1}$, which contain $\Theta^0_{n+1}$, $\Theta^1_{n+1}$ and  $\Theta^2_{n+1}$. In Figs.~\ref{mnfo0},~\ref{mnfo1}, and~\ref{mnfo2}, only the hub vertices $X_{n}^{'(\theta)}$ and $Y_{n}^{'(\theta)}$ of $\mathcal{G}_{n}^{'(\theta)}$, $\theta=1,2,3,4$, forming $\mathcal{G}'_{n+1}$
 are shown explicitly, with filled circles denoting covered vertices and empty circles denoting vacant vertices. Note that in Fig.~\ref{mnfo2},  the iterative  edge linking the  two hub vertices  $X'_{n+1}$ and $Y'_{n+1}$ of $\mathcal{G}'_{n+1}$ will  be included in the matching if both of the two hub vertices of $\mathcal{G}'_{n+1}$ are vacant after joining process.  From Figs.~\ref{mnfo0},~\ref{mnfo1}, and~\ref{mnfo2}, we  establish recursive relations governing  $\beta^0_n$, $\beta^1_n$, and $\beta^2_n$:
\begin{align} 
\beta^0_{n+1}&={\rm max}\{4\beta^0_n ,3\beta^0_n+\beta^1_n ,2\beta^0_n+2\beta^1_n\}\,, \label{beta0g}\\
\beta^1_{n+1}&={\rm max}\{3\beta^0_n+\beta^1_n ,3\beta^0_n+\beta^2_n ,2\beta^0_n+2\beta^1_n ,\notag\\
&\quad \quad\quad\quad
2\beta^0_n+\beta^1_n+\beta^2_n, \beta^0_n+3\beta^1_n \}\,, \label{beta1g}\\
\beta^2_{n+1}&={\rm max}\{2\beta^0_n+2\beta^1_n ,2\beta^0_n+\beta^1_n+\beta^2_n, \beta^0_n+3\beta^1_n , \notag\\
&\quad \quad\quad\quad
 2\beta^0_n+2\beta^2_n,\beta^0_n+2\beta^1_n+\beta^2_n,4\beta^1_n,4\beta^0_n+1,\notag\\
&\quad \quad\quad\quad
3\beta^0_n+\beta^1_n +1,2\beta^0_n+2\beta^1_n+1\}\,. \label{beta2g}
\end{align}
With initial condition $\beta^0_1=0$, $\beta^0_1=1$, and $\beta^0_1=2$, the above equations are solved to yield $\beta^0_n=\frac{2^{2n-1}-2}{3}$, $\beta^1_n=\frac{2^{2n-1}+1}{3}$, and $\beta^2_n=\frac{2^{2n-1}+4}{3}$.
\end{proof}

\begin{figure}
    \centering
    \begin{minipage}[c]{0.5\textwidth}
        \centering
        \includegraphics[width=0.27\textwidth]{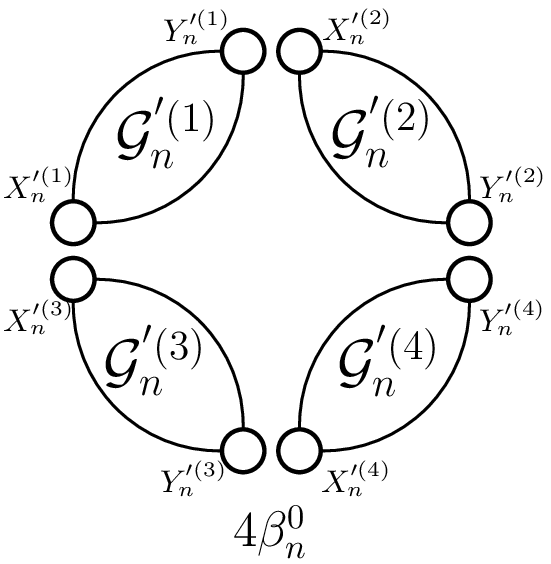}
        \includegraphics[width=0.3\textwidth]{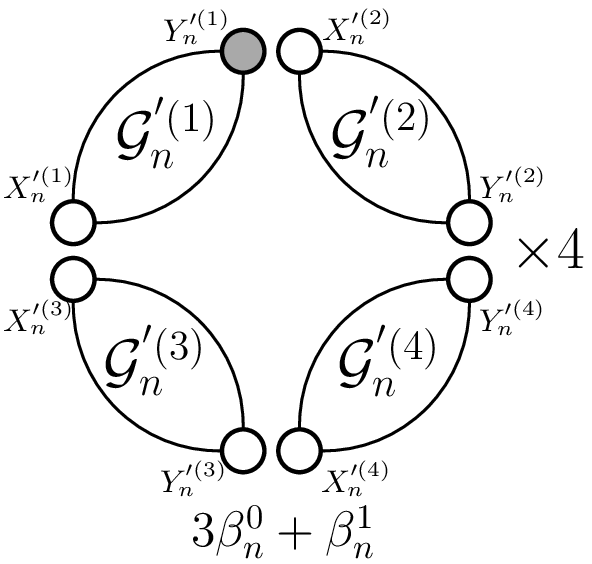}
        \includegraphics[width=0.3\textwidth]{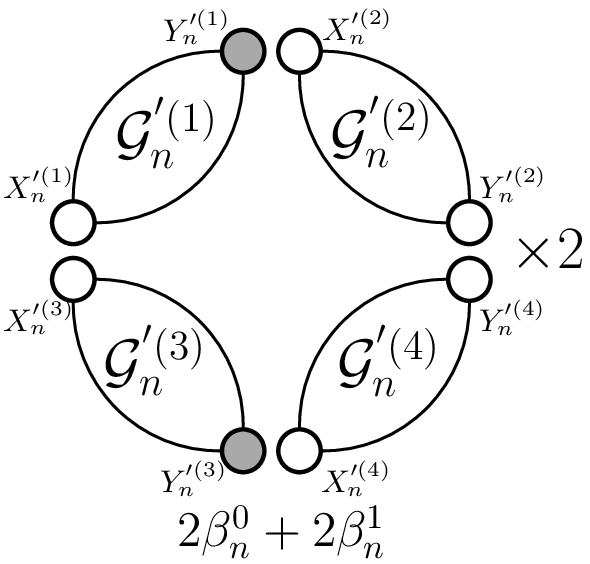}
        \\ \hspace*{\fill} \\
    \end{minipage}
    \begin{minipage}[c]{0.5\textwidth}
        \centering
        \includegraphics[width=0.3\textwidth]{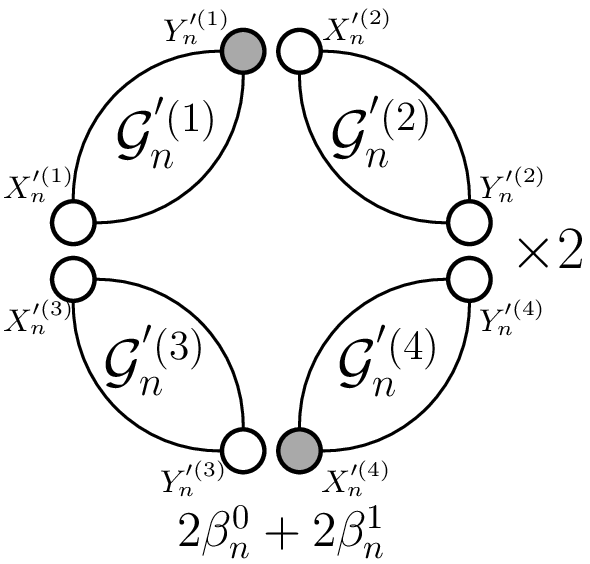}
    \end{minipage}
\caption{Illustration of all possible configurations and their sizes of matchings for graph $\mathcal{G}'_{n+1}$ belonging to $\Omega^0_{n+1}$, which contain all matchings in $\Theta^0_{n+1}$. }
\label{mnfo0}
\end{figure}

\begin{figure}
    \centering
    \begin{minipage}[c]{0.5\textwidth}
        \centering
        \includegraphics[width=0.3\textwidth]{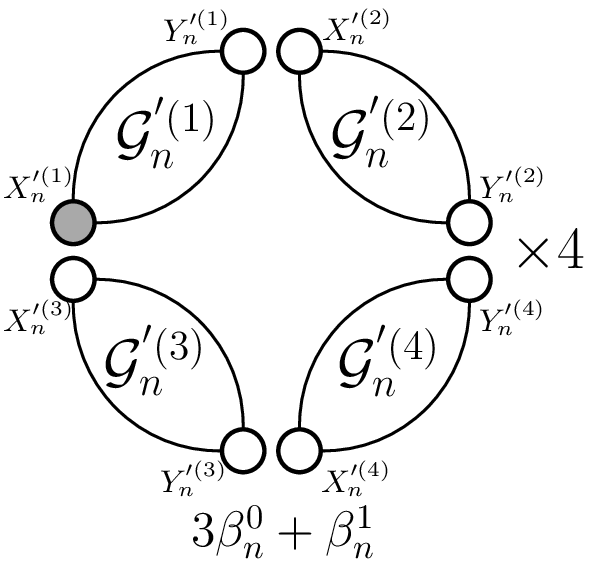}
        \includegraphics[width=0.3\textwidth]{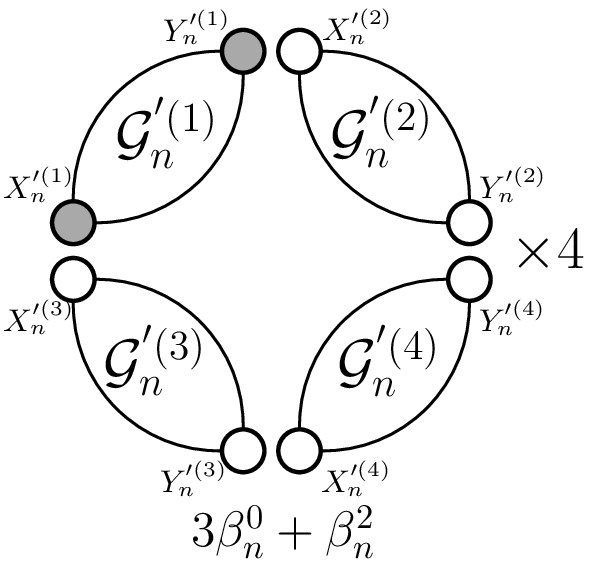}
        \includegraphics[width=0.3\textwidth]{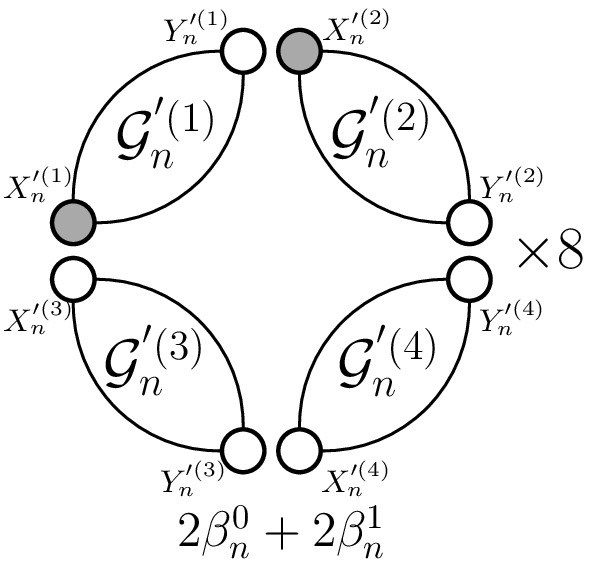}
        \\ \hspace*{\fill} \\
    \end{minipage}
    \begin{minipage}[c]{0.5\textwidth}
        \centering
        \includegraphics[width=0.3\textwidth]{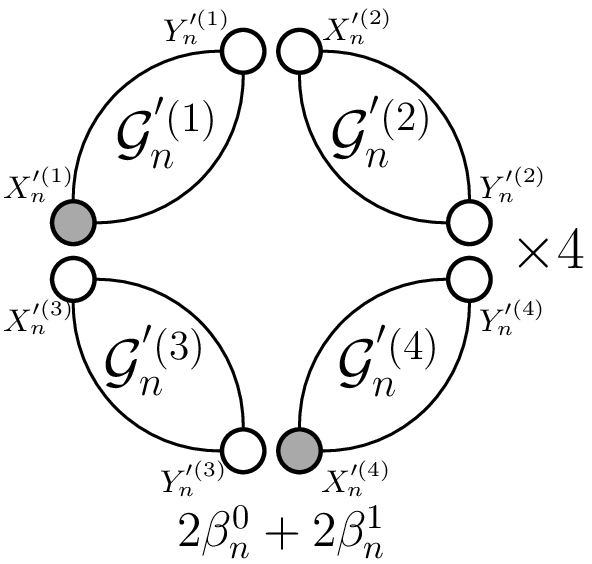}
        \includegraphics[width=0.3\textwidth]{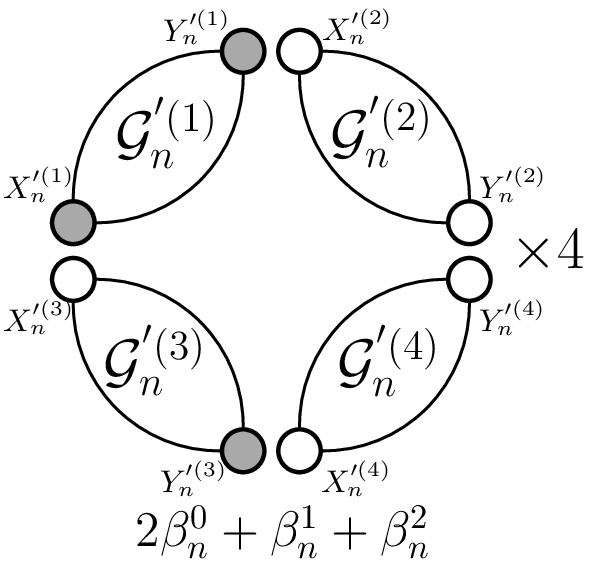}
        \includegraphics[width=0.3\textwidth]{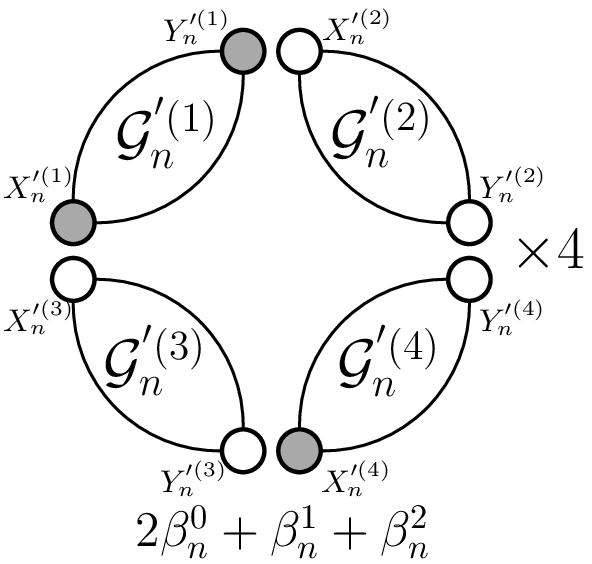}
        \\ \hspace*{\fill} \\
    \end{minipage}
    \begin{minipage}[c]{0.5\textwidth}
        \centering
        \includegraphics[width=0.3\textwidth]{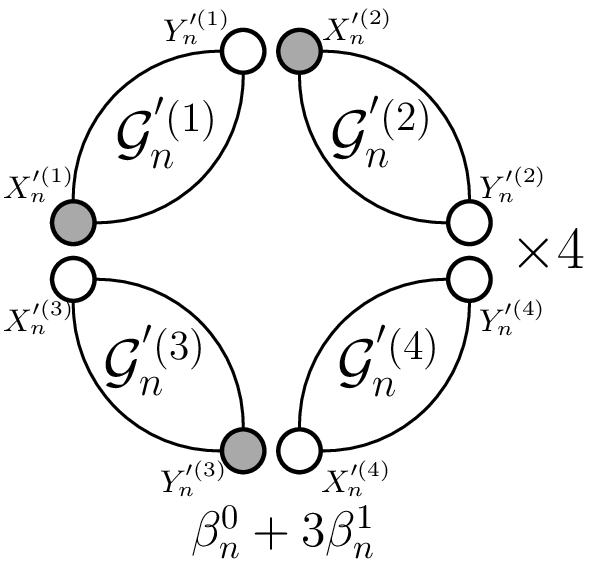}
    \end{minipage}
\caption{Illustration of all possible configurations and their sizes of matchings for graph $\mathcal{G}'_{n+1}$ belonging to $\Omega^1_{n+1}$, which contain all matchings in $\Theta^1_{n+1}$.  }
\label{mnfo1}
\end{figure}

\begin{figure}
    \centering
    \begin{minipage}[c]{0.5\textwidth}
        \centering
        \includegraphics[width=0.3\textwidth]{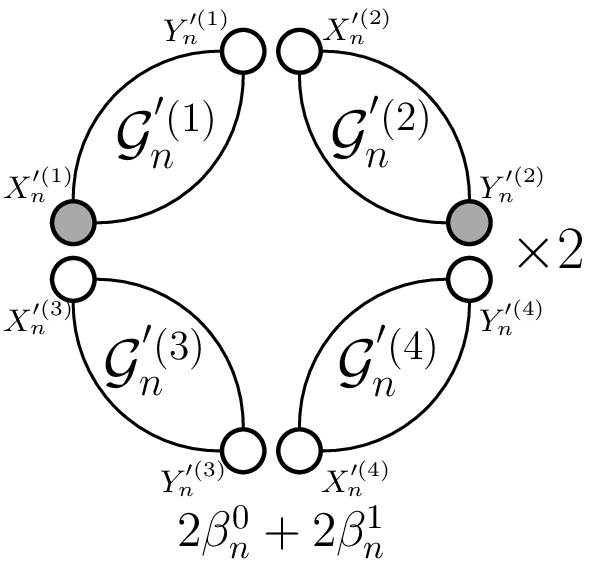}
        \includegraphics[width=0.3\textwidth]{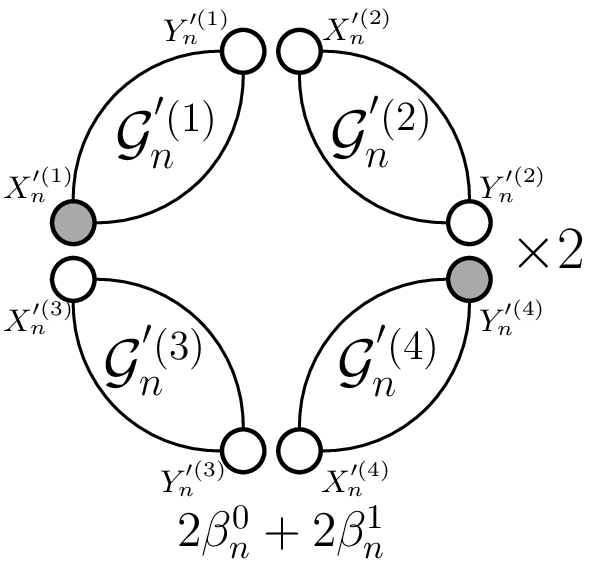}
        \includegraphics[width=0.3\textwidth]{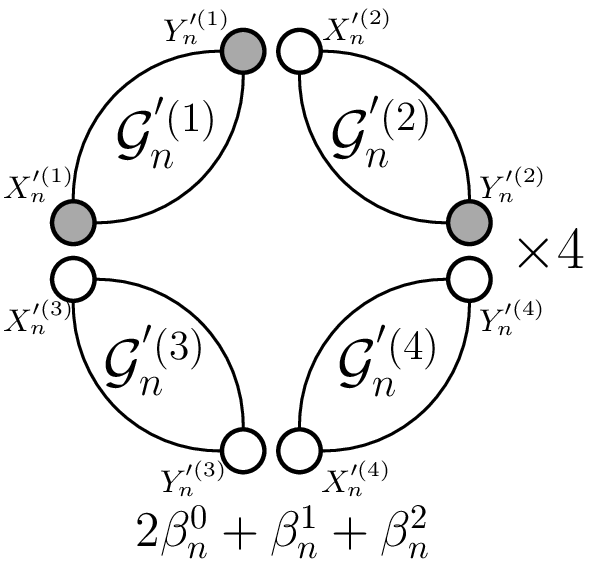}
        \\ \hspace*{\fill} \\
    \end{minipage}
    \begin{minipage}[c]{0.5\textwidth}
        \centering
        \includegraphics[width=0.3\textwidth]{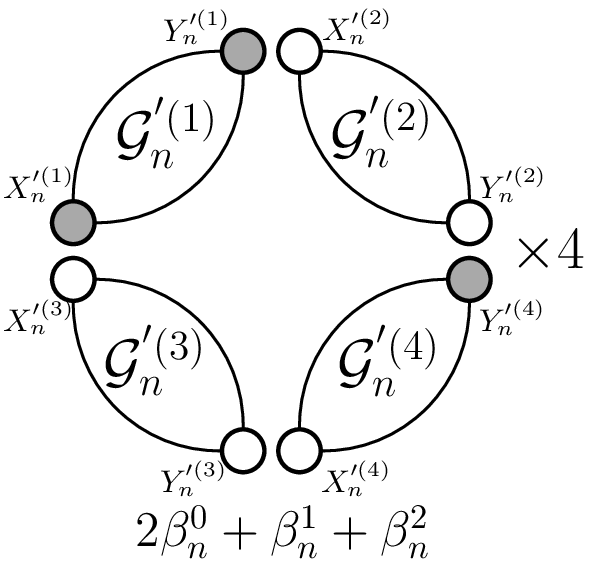}
        \includegraphics[width=0.3\textwidth]{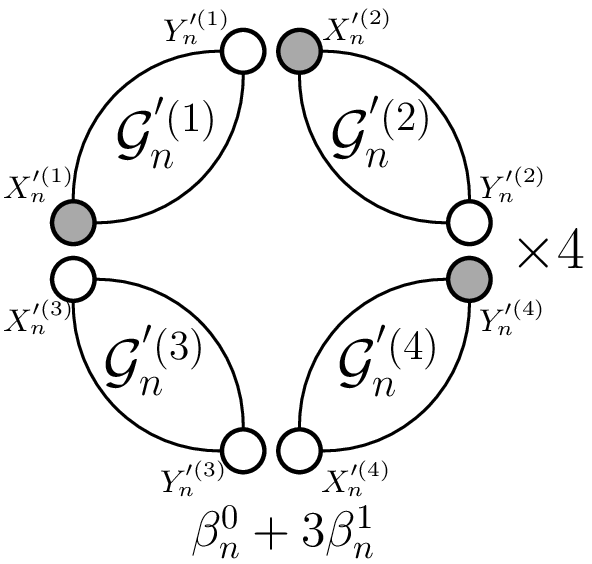}
        \includegraphics[width=0.3\textwidth]{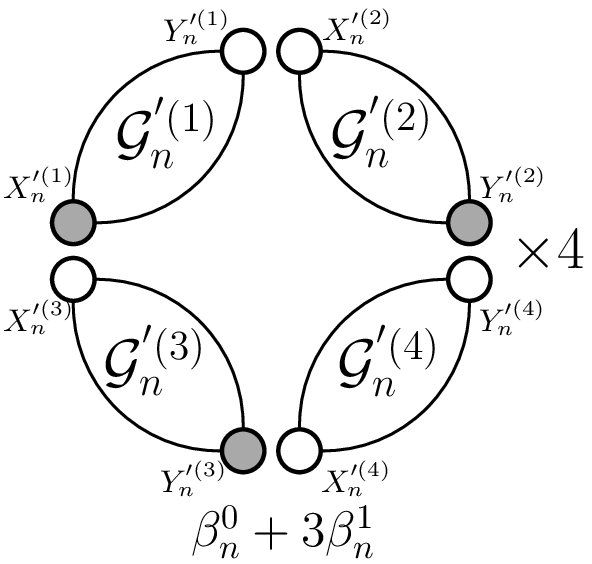}
        \\ \hspace*{\fill} \\
    \end{minipage}
    \begin{minipage}[c]{0.5\textwidth}
        \centering
        \includegraphics[width=0.3\textwidth]{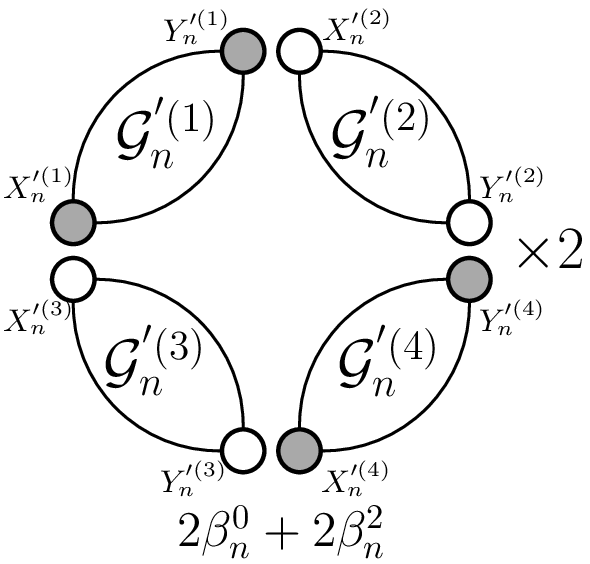}
        \includegraphics[width=0.3\textwidth]{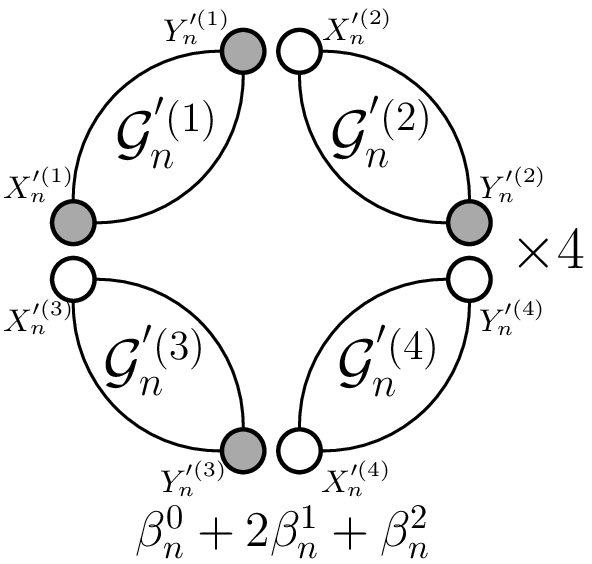}
        \includegraphics[width=0.3\textwidth]{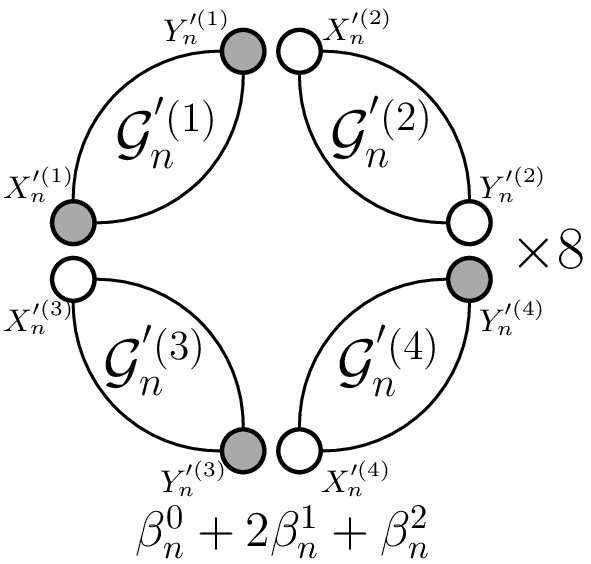}
        \\ \hspace*{\fill} \\
    \end{minipage}
    \begin{minipage}[c]{0.5\textwidth}
        \centering
        \includegraphics[width=0.3\textwidth]{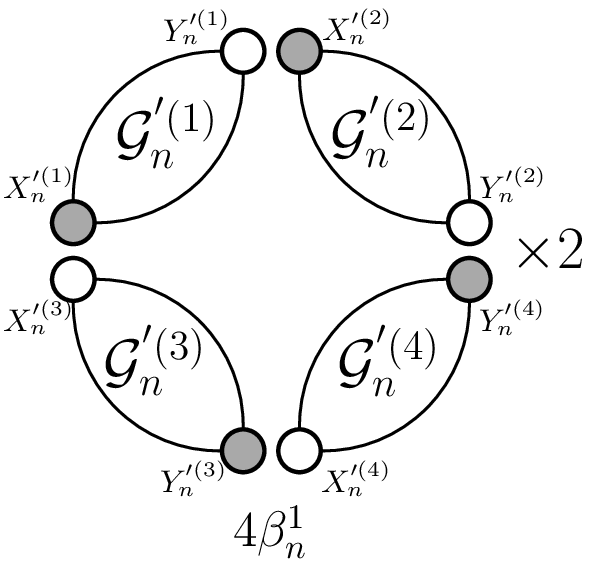}
        \includegraphics[width=0.3\textwidth]{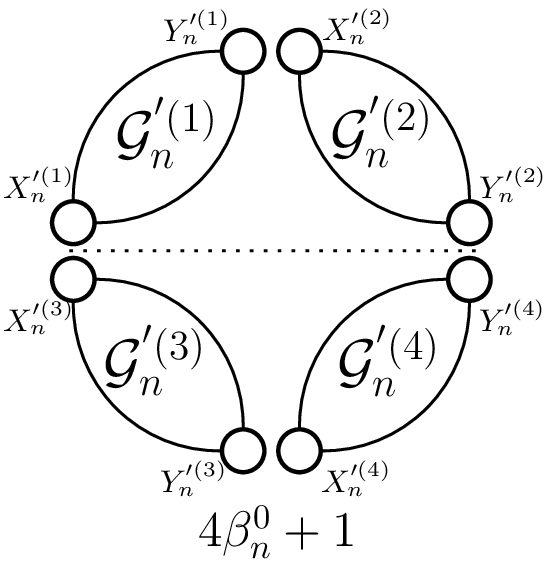}
        \includegraphics[width=0.3\textwidth]{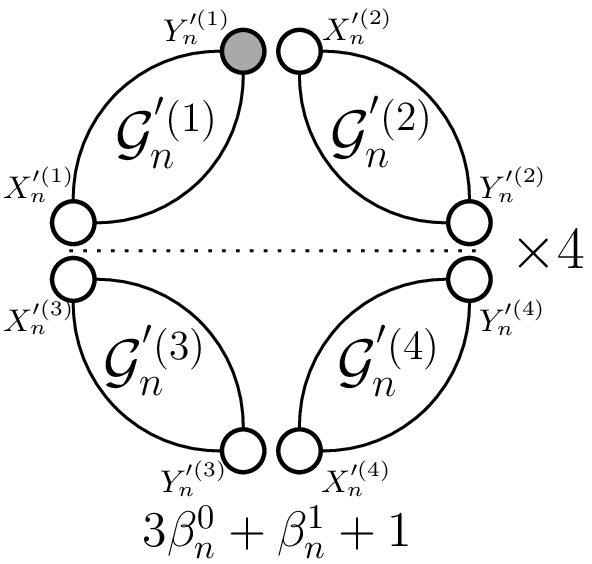}
        \\ \hspace*{\fill} \\
    \end{minipage}
    \begin{minipage}[c]{0.5\textwidth}
        \centering
        \includegraphics[width=0.3\textwidth]{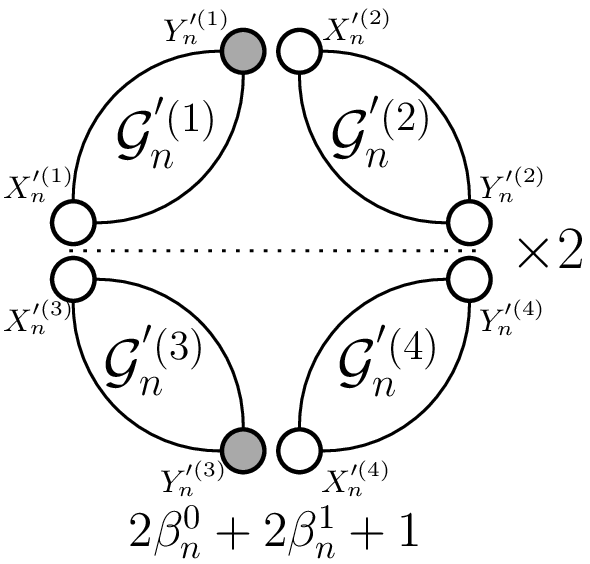}
        \includegraphics[width=0.3\textwidth]{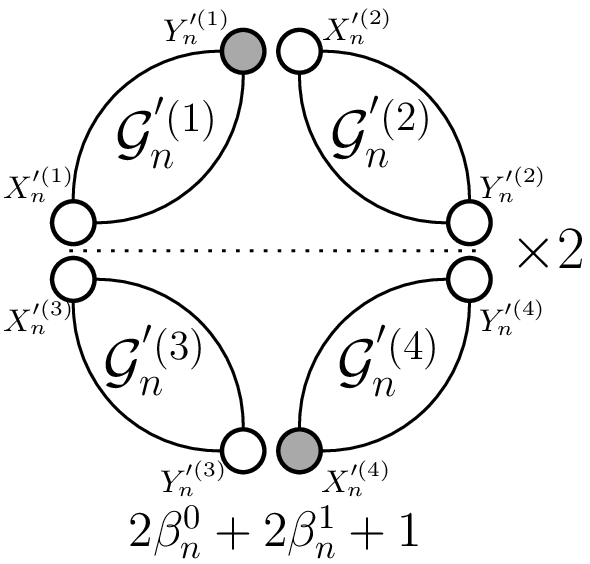}
    \end{minipage}
\caption{Illustration of all possible configurations and their sizes of matchings for graph $\mathcal{G}'_{n+1}$ belonging to $\Omega^2_{n+1}$, which contain all matchings in $\Theta^2_{n+1}$. }
\label{mnfo2}
\end{figure}

\subsubsection{Number of matchings}

Let $\theta_n$ denote the number of maximum matchings of $\mathcal{G}'_n$. To calculate $\theta_n$, we introduce two additional quantities.  Let $\phi_n$ be the number of maximum matchings in $\Omega^0_n$, and let $\varphi_n$ be the number   of maximum matchings in $\Omega^1_n$. For small $n$, quantities $\phi_n$, $\varphi_n$ and $\theta_n$ can be easily determined by using a computer. For example, for $n=1$, $\theta_1 = 2$, $\phi_1 = 1$ and $\varphi_1=2$. For large $n$, they can be determined recursively as follows.

\begin{theorem}\label{mnnot}
For graph $\mathcal{G}_n$, $n\geq1$, the three quantities $\theta_n$, $\phi_n$ and $\varphi_n$ can be calculated recursively according to the following relations:
\begin{align}
\theta_{n+1}=&2\theta_n^2\phi_n^2+2\varphi_n^4+12\theta_n\phi_n\varphi_n^2, \label{thetag} \\
\phi_{n+1}=&4\phi_n^2\varphi_n^2, \label{phig} \\
\varphi_{n+1}=&4\theta_n\phi_n^2\varphi_n+4\phi_n\varphi_n^3. \label{varphig}
\end{align}
with initial conditions $\theta_1 = 2$, $\phi_1 = 1$, and $\varphi_1=2$.
\end{theorem}
\begin{proof}
We only prove Eq.~\eqref{phig},  because the other two equations can be proved analogously. Since for $n\geq 1$, $\beta^0_n<\beta^1_n<\beta^2_n$, according to  Eq.~\eqref{beta0g} and Fig.~\ref{mnfo0},  all maximum matchings with size $\beta^0_{n+1}$ in $\Omega^0_{n+1}$ are those matchings having size $2\beta^0_n+2\beta^1_n$, which together with the symmetry of graph $\mathcal{G}'_{n+1}$, yields Eq.~\eqref{phig}.
\end{proof}


\section{Independence number and the number of maximum independence sets}

In this section, we study the independence number and the number of  MISs in the two studied self-similar scale-free networks.

\subsection{Independence number and the number of maximum independence sets in fractal scale-free networks}

We first study the independence number and the number of  MISs in fractal scale-free graph  $\mathcal{G}_{n}$.

\subsubsection{Independence number}

Let $\alpha_n$ denote   independence number  of graph $\mathcal{G}_n$.  To find   $\alpha_n$,  we  define some intermediate quantities. Note that all the independent sets of $\mathcal{G}_n$ can be classified into three types: $\Psi^0_n$, $\Psi^1_n$ and $\Psi^2_n$, where $\Psi^k_n$, $k=0,1,2$, represent the set of independent sets, each  including exactly $k$ initial vertices  of $\mathcal{G}_n$. Let $\Phi^k_n$, $k=0,1,2$, be the subset of $\Psi^k_n$, where each independent set has the largest cardinality, denoted as $\alpha^k_n$, $k=0,1,2$. Then, $\alpha_n$ can be represented  as $\alpha_n={\rm max}\{\alpha^0_n, \alpha^1_n,\alpha^2_n\}$.

\begin{theorem}\label{InnumF}
The independence number of graph $\mathcal{G}_n$, $n\geq 1$, is $\alpha_n=2^{2n-2}$.
\end{theorem}
\begin{proof}
Since $\alpha_n={\rm max}\{\alpha^0_n, \alpha^1_n,\alpha^2_n\}$,  the problem of determining  $\alpha_n$  is reduced to evaluating  the three quantities $\alpha^0_n$, $\alpha^1_n$ and $\alpha^2_n$.  By using the self-similar structure,  it is not difficult to prove that quantities $\alpha^0_n$, $\alpha^1_n$ and $\alpha^2_n$ satisfy the following relations:
\begin{align} 
\alpha^0_{n+1}&={\rm max}\{4\alpha^0_n ,2\alpha^0_n+2\alpha^1_n-1\}\,, \label{alpha0F}\\
\alpha^1_{n+1}&={\rm max}\{2\alpha^0_n+2\alpha^1_n-1 ,\alpha^0_n+2\alpha^1_n+\alpha^2_n-2\}\,, \label{alpha1F}\\
\alpha^2_{n+1}&={\rm max}\{2\alpha^1_n+2\alpha^2_n-3 ,4\alpha^1_n-2\}\,.\label{alpha2F}
\end{align}

By definition, $\alpha_{n+1}^k$, $k = 0,1,2$, is the cardinality of an independent set in  $\Psi^k_n$. Below, we will show that  $\Psi^0_{n+1}$, $\Psi^1_{n+1}$, and $\Psi^2_{n+1}$ can be iteratively constructed from $\Psi^0_{n}$, $\Psi^1_{n}$, and $\Psi^2_{n}$. Thus, $\alpha_{n+1}^0$,  $\alpha_{n+1}^1$, and $\alpha_{n+1}^0$ can be expressed in terms of $\alpha_n^0$, $\alpha_n^1$, and $\alpha_n^2$. We now prove  graphically the above recursive relations given by Eqs.~\eqref{alpha0F},~\eqref{alpha1F}, and~\eqref{alpha2F}.

We first prove Eq.~\eqref{alpha0F}.  By the second construction,  $\mathcal{G}_{n+1}$ consists of four copies of $\mathcal{G}_n$, $\mathcal{G}_{n}^{(\theta)}$, $\theta=1,2,3,4$. By definition, for any independent  set $\chi$ in $\Psi_{n+1}^0$, the two initial  vertices $X_{n+1}$ and $Y_{n+1}$  of $\mathcal{G}_{n+1}$ do not belong to $\chi$,  implying  that the corresponding two pairs ($X_{n}^{(1)}$ and $X_{n}^{(3)}$,  $Y_{n}^{(2)}$ and $Y_{n}^{(4)}$) of the identified initial vertices of $\mathcal{G}_{n}^{(\theta)}$,  $\theta=1,2,3,4$, are not in $\chi$.  In addition, since the two hub vertices $W_{n+1}$ and $Z_{n+1}$  of $\mathcal{G}_{n+1}$ are adjacent, at most one of them is in  $\chi$, meaning that among the two pairs of vertices ($Y_{n}^{(1)}$ and $X_{n}^{(2)}$,  $Y_{n}^{(3)}$ and $X_{n}^{(4)}$), at most  one pair is in  $\chi$,  see Fig.~\ref{ifo0}. Therefore, we can construct set $\chi$ only from $\Psi_n^0$ and $\Psi_n^1$ by considering whether the initial vertices of $\mathcal{G}_{n}^{(\theta)}$, $\theta=1,2,3,4$, are in  $\chi$ or not.  Figure~\ref{ifo0} illustrates all possible configurations of independent  sets in $\Phi_{n+1}^0$ that include $\Psi_{n+1}^0$ as its subset. From Fig.~\ref{ifo0}, we obtain Eq.~\eqref{alpha0F}.

Similarly we can  prove Eqs.~\eqref{alpha1F} and~\eqref{alpha2F}, the graphical representations of which are shown in Figs.~\ref{ifo1} and~\ref{ifo2}, respectively.

Considering the initial conditions $\alpha^0_1=1$, $\alpha^1_1=1$, and $\alpha^2_1=2$, Eqs.~\eqref{alpha0F},~\eqref{alpha1F}, and~\eqref{alpha2F} are solved to yield $\alpha^0_n=2^{2n-2}$, $\alpha^1_n=2^{2n-2}-2^{n-1}+1$, and $\alpha^2_n=2^{2n-2}-(n-1)2^{n-1}+1$. Thus,  $\alpha_n={\rm max}\{\alpha^0_n, \alpha^1_n,\alpha^2_n\}=2^{2n-2}$.
\end{proof}

\begin{figure}
\centering
    \begin{minipage}[c]{0.5\textwidth}
        \centering
        \includegraphics[width=0.27\textwidth]{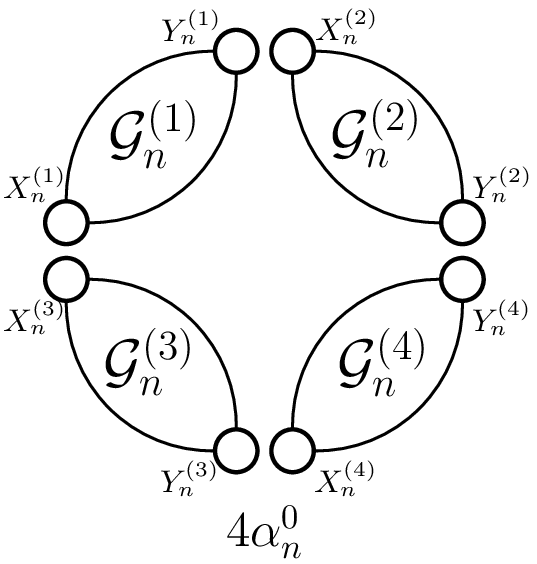}
        \includegraphics[width=0.3\textwidth]{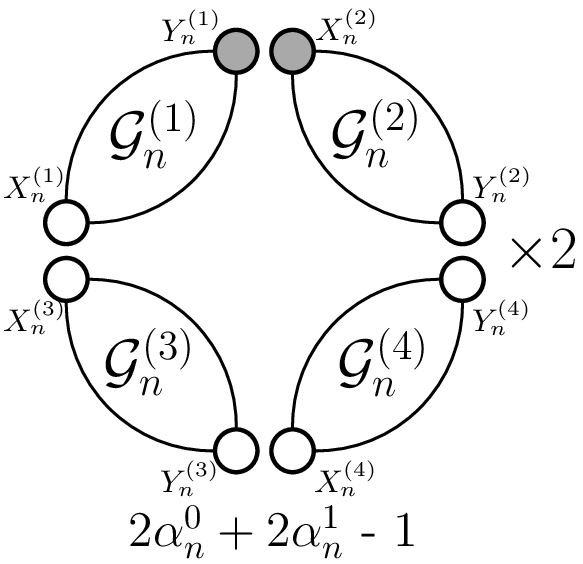}
    \end{minipage}
\caption{Illustration of all possible configurations and their sizes of independent sets $\Psi^0_{n+1}$ in graph $\mathcal{G}_{n+1}$,  which contain all independent sets in $\Phi^0_{n+1}$.}
\label{ifo0}
\end{figure}

\begin{figure}
\centering
    \begin{minipage}[c]{0.5\textwidth}
        \centering
        \includegraphics[width=0.3\textwidth]{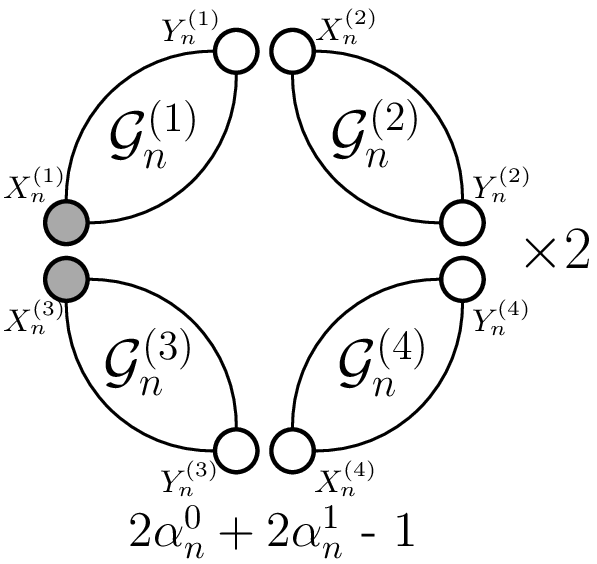}
        \includegraphics[width=0.3\textwidth]{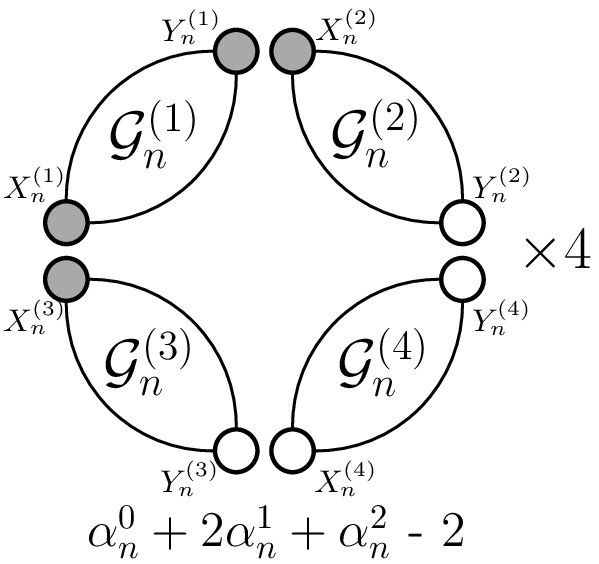}
    \end{minipage}
\caption{Illustration of all possible configurations and their sizes of independent sets $\Psi^1_{n+1}$ in graph $\mathcal{G}_{n+1}$,  which contain all independent sets in $\Phi^1_{n+1}$.}
\label{ifo1}
\end{figure}

\begin{figure}
\centering
    \begin{minipage}[c]{0.5\textwidth}
        \centering
        \includegraphics[width=0.27\textwidth]{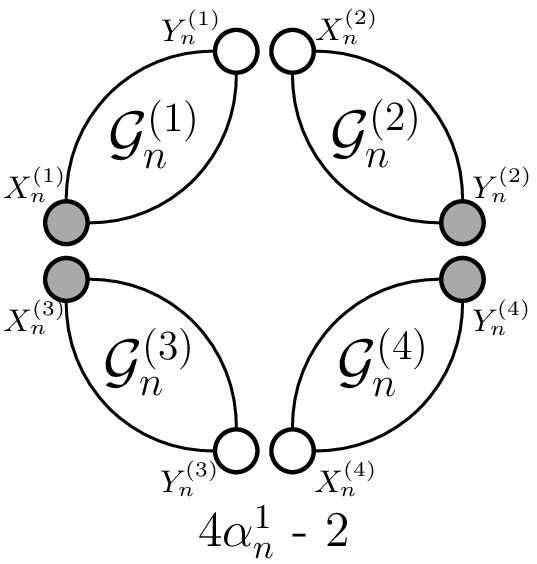}
        \includegraphics[width=0.3\textwidth]{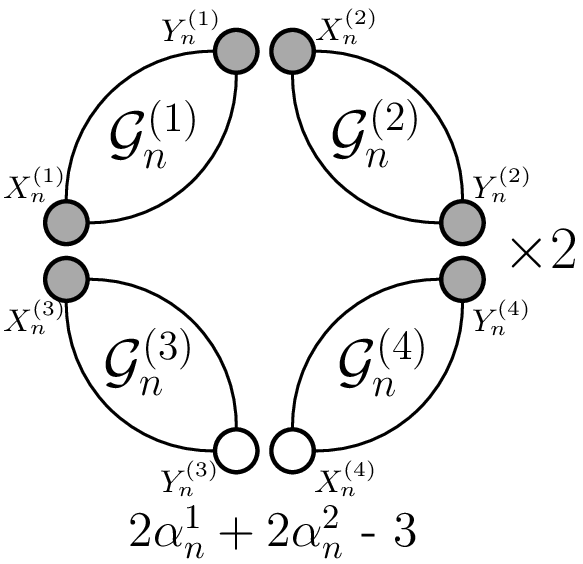}
    \end{minipage}
\caption{Illustration of all possible configurations and their sizes of independent sets $\Psi^2_{n+1}$ in graph $\mathcal{G}_{n+1}$,  which contain all independent sets in $\Phi^2_{n+1}$.}
\label{ifo2}
\end{figure}

\subsubsection{Number of maximum independent sets}

In addition the  independence  number,  the number of MISs in graph $\mathcal{G}_n$ can also be determined exactly.

\begin{theorem}\label{NubMISF}
The number of maximum independent sets in graph $\mathcal{G}_n$, $n \geq 1$, is $2^{2^{2n-2}}$.
\end{theorem}
\begin{proof}
From the proof of Theorem~\ref{InnumF}, we have $\alpha^0_{n}> \alpha^1_n > \alpha^2_n$ when $n \geq 2$. Thus, according to Eq.~\eqref{alpha0F}, $\alpha^0_{n+1}={\rm max}\{4\alpha^0_n, 2\alpha^0_n+2\alpha^1_n-1\}=4\alpha^0_n$. Let $x_n$ denote the number of  MISs in graph $\mathcal{G}_n$. From Fig.~\ref{ifo0},  we obtain
\begin{align}\label{recd}
x_{n+1}=x_{n}^4,
\end{align}
which, under the initial value $x_1=2$, is solved  to yield $x_n=2^{2^{2n-2}}$.
\end{proof}

Theorem \ref{NubMISF} shows that number of  MISs in graph $\mathcal{G}_n$ grows exponentially with the number of vertices $N_n$ .

\subsection{Independence number and the number of maximum independent sets in non-fractal scale-free networks}

We continue to study the independence number and the number of  MISs in non-fractal scale-free graph $\mathcal{G}'_n$.

\subsubsection{Independence number}

We classify  all the independent sets of $\mathcal{G}'_n$  into two types: $\Psi^0_n$ and $\Psi^1_n$, where $\Psi^k_n$, $k=0,1,2$, represent the set of independent sets, each  including exactly $k$ hub vertices  of $\mathcal{G}'_n$. Let $\Phi^k_n$, $k=0,1,2$, be the subset of $\Psi^k_n$, where each independent set has the largest cardinality, denoted by $\alpha^k_n$, $k=0,1,2$. Since there is an edge connecting the two hub vertices $X'_n$ and $Y'_n$, set  $\Psi^2_n$ is empty, implying $\alpha^2_n=0$. Let  $\alpha_n$ denote the independence number of $\mathcal{G}'_n$. Then,  $\alpha_n$  can  be expressed  as $\alpha_n={\rm max}\{\alpha^0_n, \alpha^1_n\}$.

\begin{theorem}
The independence number of graph $\mathcal{G}'_n$, $n\geq 1$, is $\alpha_n=2^{2n-1}$.
\end{theorem}
\begin{proof}
Considering $\alpha_n={\rm max}\{\alpha^0_n, \alpha^1_n\}$,  in order to determine, we  alternatively evaluate the two quantities $\alpha^0_n$ and $\alpha^1_n$ by using the self-similarity of the graph. First, we show that $\alpha^0_n$ and $\alpha^1_n$ obey the following recursion relations:
\begin{align} 
\alpha^0_{n+1}&={\rm max}\{4\alpha^0_n ,2\alpha^0_n+2\alpha^1_n-1,4\alpha^1_n-2\}\,, \label{alpha0G}\\
\alpha^1_{n+1}&=2\alpha^0_n+2\alpha^1_n-1\,. \label{alpha1G}
\end{align}
Equations~\eqref{alpha0G} and~\eqref{alpha1G} be proved graphically. Figures~\ref{info0}, and~\ref{info1} show the graphical representations  of Eqs.~\eqref{alpha0G} and~\eqref{alpha1G}, respectively.

Using the initial conditions $\alpha^0_1=2$ and $\alpha^1_1=1$, Eqs.~\eqref{alpha0G} and~\eqref{alpha1G} are solved to yield exact solutions for $\alpha^0_n$ and $\alpha^1_n$ as $\alpha^0_n=2^{2n-1}$ and $\alpha^1_n=2^{2n-1}-2^{n}+1$. Then, we have $\alpha_n={\rm max}\{\alpha^0_n, \alpha^1_n\}=2^{2n-1}$ for $n\geq 1$.
\end{proof}

\begin{figure}
\centering
    \begin{minipage}[c]{0.5\textwidth}
        \centering
        \includegraphics[width=0.27\textwidth]{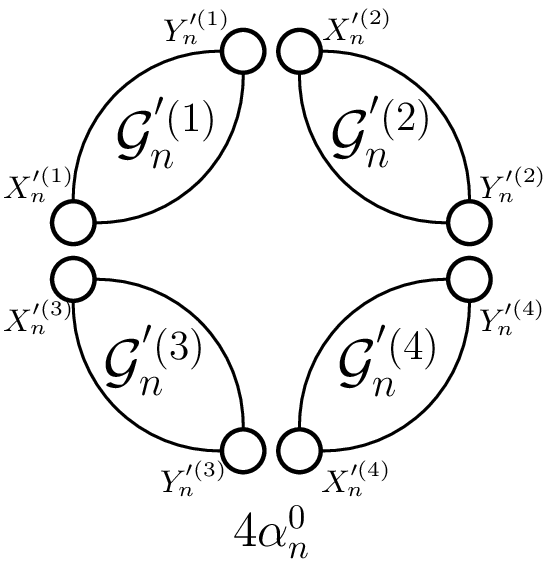}
        \includegraphics[width=0.3\textwidth]{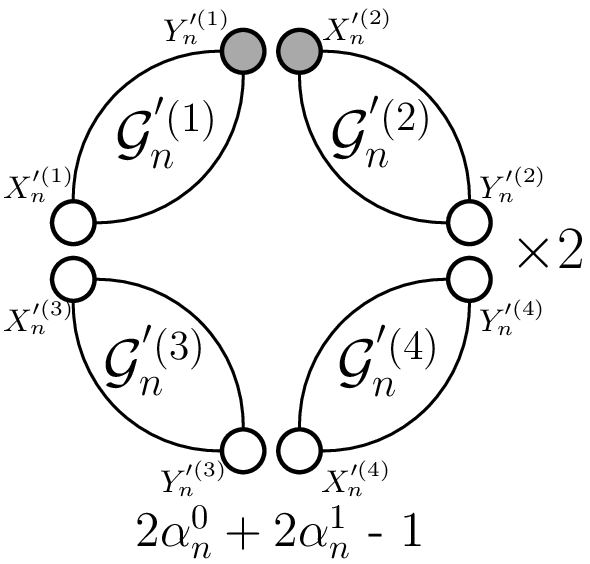}
        \includegraphics[width=0.3\textwidth]{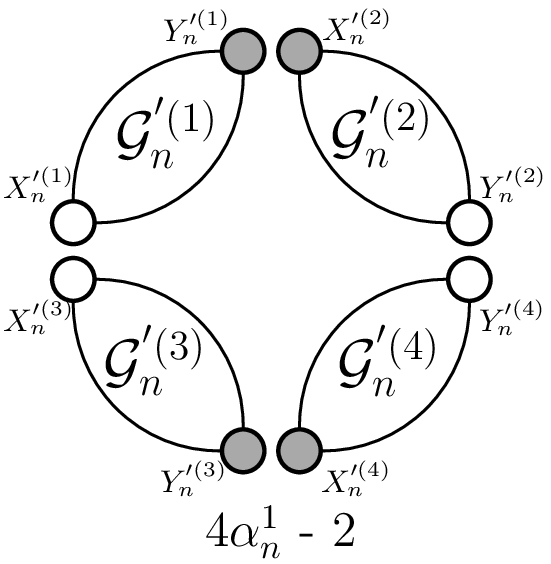}
    \end{minipage}
\caption{Illustration of all possible configurations and their sizes of independent sets $\Psi^0_{n+1}$ in graph $\mathcal{G}'_{n+1}$,  which contain all independent sets in $\Phi^0_{n+1}$.}
\label{info0}
\end{figure}

\begin{figure}
\centering
    \begin{minipage}[c]{0.5\textwidth}
        \centering
        \includegraphics[width=0.3\textwidth]{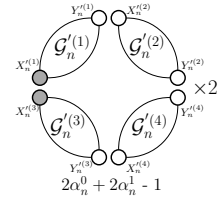}
    \end{minipage}
\caption{Illustration of all possible configurations and their sizes of independent sets $\Psi^1_{n+1}$ in graph $\mathcal{G}'_{n+1}$,  which contain all independent sets in $\Phi^1_{n+1}$.}
\label{info1}
\end{figure}

\subsubsection{Number of maximum independence sets}

In contrast to the its fractal counterpart  $\mathcal{G}_{n}$,  the non-fractal scale-free graph $\mathcal{G'}_{n}$ has only one maximum independence set for all $n \geq 1$.

\begin{theorem}\label{NumMISg}
In the non-fractal  scale-free graph $\mathcal{G'}_{n}$, $n \geq 1$, there exists a unique maximum independence set.
\end{theorem}
\begin{proof}
Let $x_n$ denote the number of MISs in $\mathcal{G}'_n$. Equation~\eqref{alpha0G}  and Fig.~\ref{info0}  show that for $n \geq 2$ any MIS of $\mathcal{G}'_{n+1}$ is in fact the union of maximum independent sets in $\Phi_n^0$, of the four copies of $\mathcal{G}'_{n}$ (i.e. $\mathcal{G}_{n}^{'(1)}$, $\mathcal{G}_{n}^{'(2)}$, $\mathcal{G}_{n}^{'(3)}$, and $\mathcal{G}_{n}^{'(4)}$) forming $\mathcal{G}'_{n+1}$. Thus, any MIS of $\mathcal{G}_{n+1}$ is determined by those of $\mathcal{G}_{n}^{'(1)}$, $\mathcal{G}_{n}^{'(2)}$, $\mathcal{G}_{n}^{'(3)}$, and $\mathcal{G}_{n}^{'(4)}$. Moreover, $x_{n+1}=x_{n}^4$. Since $x_{1}=1$, we have $x_{n}=1$ for  all $n \geq 1$.
\end{proof}

Theorem~\ref{NumMISg} indicates that  for all $n \geq 1$,  $\mathcal{G}'_n$  is a unique independence graph. Furthermore,  it is easy to see that the unique MIS of $\mathcal{G}'_n$, $n \geq 1$, contains exactly  all the  vertices with degree two that are generated at iteration  $n-1$.

\section{Domination number and the number of minimum dominating sets}

In this section, we study the domination number and the number of MDSs in two self-similar scale-free networks $\mathcal{G}_n$ and $\mathcal{G}'_n$.

\subsection{Domination number and the number of  MDSs in fractal scale-free networks}

We first study the domination number and the number of MDSs in the fractal scale-free network $\mathcal{G}_n$.

Let $\gamma_n$ denote the domination number of graph $\mathcal{G}_n$.  In order to determine  $\gamma_n$, we classify  into all the dominating sets of $\mathcal{G}_n$  into three groups: $\Gamma^0_n$, $\Gamma^1_n$ and $\Gamma^2_n$,  where $\Gamma^k_n$, $k=0,1,2$, represent the set of those dominating sets including exactly $k$ initial vertices  of $\mathcal{G}_n$. Moreover, let $\Upsilon^k_n$, $k=0,1,2$, be the subset of $\Upsilon^k_n$, where each independent set has the largest cardinality, denoted as $\gamma^k_n$, $k=0,1,2$.
By definition, we have $\gamma_n={\rm min}\{\gamma^0_n, \gamma^1_n, \gamma^2_n\}$.

\begin{theorem}
For $n \geq 2$, the domination number of graph $\mathcal{G}_n$ is  $\gamma_n=\frac{5\cdot2^{2n-4}+4}{3}$.
\end{theorem}
\begin{proof}
Since the problem of determining $\gamma_n$ can be reduced to finding $\gamma^0_n$, $\gamma^1_n$ and $\gamma^2_n$, we now determine these three intermediate quantities. To this end, we provide the following recursion relation for $n\geq 2$ governing these quantities:
\begin{align} 
\gamma^0_{n+1}&={\rm min}\{4\gamma^0_n,2\gamma^0_n+2\gamma^1_n-1,4\gamma^1_n-2\}\,, \label{gamma0F}\\
\gamma^1_{n+1}&={\rm min}\{2\gamma^0_n+2\gamma^1_n-1,\gamma^0_n+2\gamma^1_n+\gamma^1_n-2,\notag\\
&\quad \quad\quad\quad
2\gamma^1_n+2\gamma^2_n-3\}\,,  \label{gamma1F}\\
\gamma^2_{n+1}&={\rm min}\{4\gamma^1_n-2,2\gamma^1_n+2\gamma^2_n-3,4\gamma^2_n-4\}\,. \label{gamma2F}
\end{align}

Equations~\eqref{gamma0F},~\eqref{gamma1F}, and~\eqref{gamma2F}  can all be proved graphically.

We first prove Eq.~\eqref{gamma0F}.  According to Fig.~\ref{fcons2}, $\mathcal{G}_{n+1}$ is consist of four copies of $\mathcal{G}_n$, $\mathcal{G}_{n}^{(\theta)}$, $\theta=1,2,3,4$. By definition, for any dominating  set $\xi$ in $\Upsilon_{n+1}^0$, both of the two initial  vertices $X_{n+1}$ and $Y_{n+1}$  of $\mathcal{G}_{n+1}$ are not in $\xi$,  implying  that the corresponding two pairs ($X_{n}^{(1)}$ and $X_{n}^{(3)}$,  $Y_{n}^{(2)}$ and $Y_{n}^{(4)}$) of identified initial vertices of $\mathcal{G}_{n}^{(\theta)}$,  $\theta=1,2,3,4$, are not in $\xi$. In addition, according to the number of   hub vertices in a dominating  sets belonging  $\Upsilon_{n+1}^0$, the dominating  sets in $\Upsilon_{n+1}^0$ can be further sorted into three disjoint subsets.  Figure~\ref{dfo0} illustrates all possible configurations of dominating sets in $\Gamma_{n+1}^0$ that contains all dominating sets in $\Upsilon_{n+1}^0$. In Fig.~\ref{dfo0}, only the initial vertices of $\mathcal{G}_{n}^{(\theta)}$, $\theta=1,2,3,4$, are shown, with solid vertices being in the dominating sets, while open vertices not. From Fig.~\ref{dfo0}, we establish Eq.~\eqref{gamma0F}.

In a similar way, we can prove Eqs.~\eqref{gamma1F} and~\eqref{gamma2F}, the graphical representations of which are provided in Figs.~\ref{dfo1} and~\ref{dfo2},  respectively.

With initial condition $\gamma^0_2=4$, $\gamma^1_2=3$ and $\gamma^2_2=3$, Eqs.~\eqref{gamma0F},~\eqref{gamma1F}, and~\eqref{gamma2F} are solved to yield $\gamma^0_n=\frac{5\cdot2^{2n-4}+3\cdot2^{n-1}-2}{3}$, $\gamma^1_n=\frac{5\cdot2^{2n-4}+3\cdot2^{n-2}+1}{3}$, and $\gamma^2_n=\frac{5\cdot2^{2n-4}+4}{3}$.  Thus,  the domination number of graph $\mathcal{G}_n$ is  $\gamma_n={\rm min}\{\gamma^0_n, \gamma^1_n, \gamma^2_n\}=\gamma^2_n=\frac{5\cdot2^{2n-4}+4}{3}$ for all $n \geq 2$.
\end{proof}

\begin{figure}
\centering
    \begin{minipage}[c]{0.5\textwidth}
        \centering
        \includegraphics[width=0.27\textwidth]{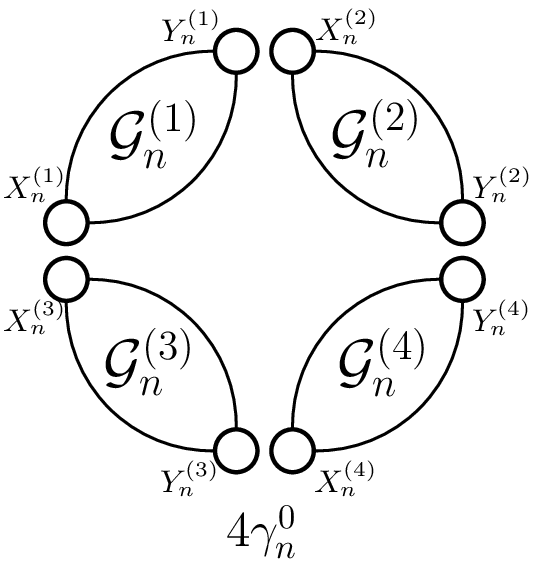}
        \includegraphics[width=0.3\textwidth]{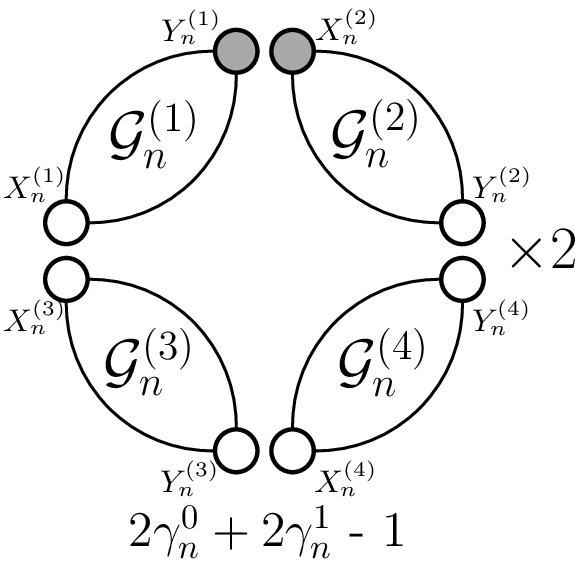}
        \includegraphics[width=0.27\textwidth]{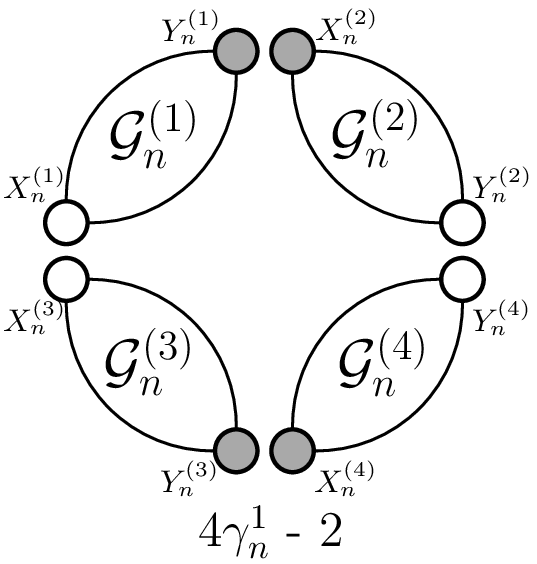}
    \end{minipage}
\caption{Illustration of all possible configurations and their sizes of dominating sets $\Gamma^0_{n+1}$ in graph $\mathcal{G}_{n+1}$ containing $\Upsilon^0_{n+1}$.}
\label{dfo0}
\end{figure}

\begin{figure}
\centering
    \begin{minipage}[c]{0.5\textwidth}
        \centering
        \includegraphics[width=0.3\textwidth]{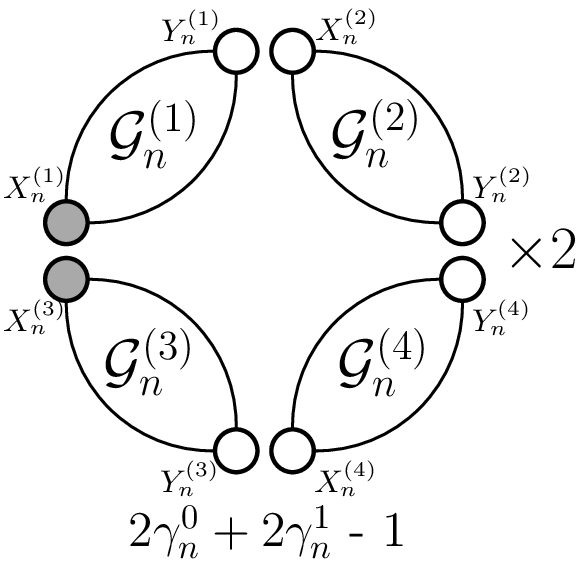}
        \includegraphics[width=0.3\textwidth]{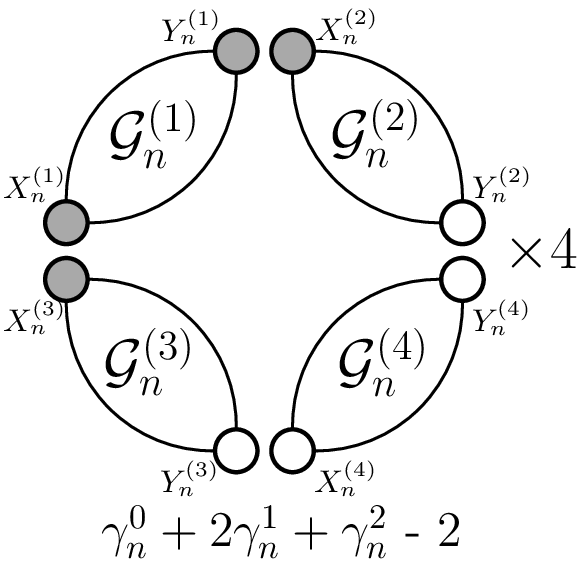}
        \includegraphics[width=0.3\textwidth]{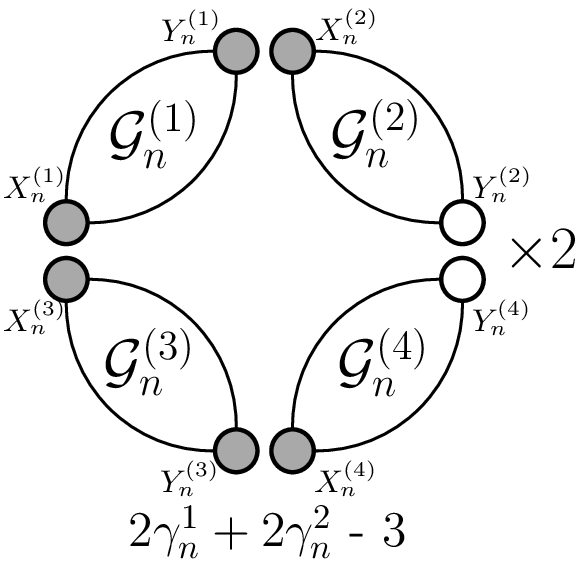}
    \end{minipage}
\caption{Illustration of all possible configurations and their sizes of dominating sets $\Gamma^1_{n+1}$ in graph $\mathcal{G}_{n+1}$ containing $\Upsilon^1_{n+1}$.}
\label{dfo1}
\end{figure}

\begin{figure}
\centering
    \begin{minipage}[c]{0.5\textwidth}
        \centering
        \includegraphics[width=0.27\textwidth]{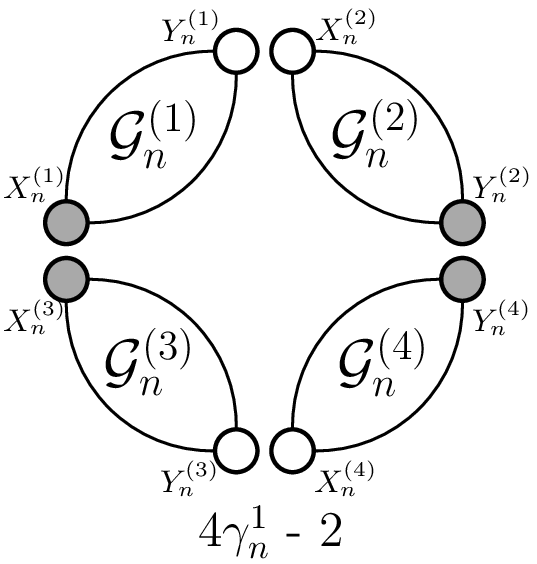}
        \includegraphics[width=0.3\textwidth]{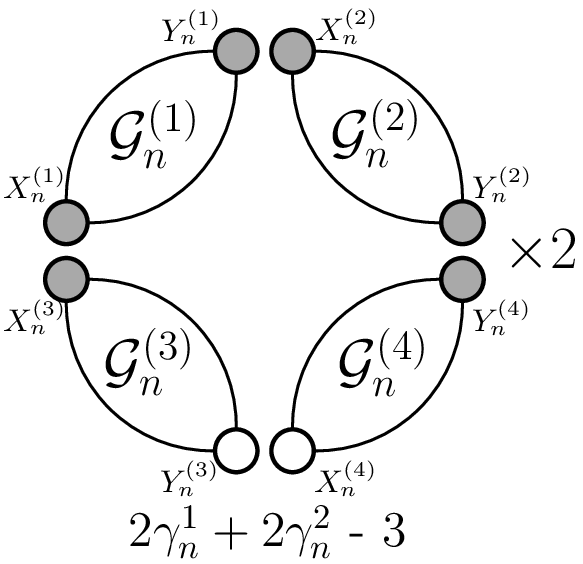}
        \includegraphics[width=0.27\textwidth]{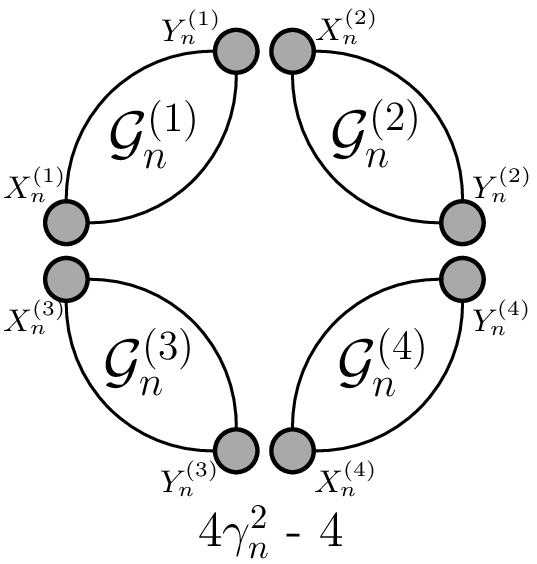}
    \end{minipage}
\caption{Illustration of all possible configurations and their sizes of dominating sets $\Gamma^2_{n+1}$ in graph $\mathcal{G}_{n+1}$ containing $\Upsilon^2_{n+1}$.}
\label{dfo2}
\end{figure}

\subsubsection{Number of minimum dominating sets}

In addition the domination  number,  the number of MDSs in graph $\mathcal{G}_n$ can also be determined exactly.

\begin{theorem}\label{NubMDSF}
The number of maximum dominating sets in graph $\mathcal{G}_n$, $n \geq 2$, is $2^{2^{2n-4}}$.
\end{theorem}
\begin{proof}
    Let $y_n$ denote the number of  MDSs in graph $\mathcal{G}_n$. From Eq.~\eqref{gamma2F} and Fig.~\ref{ifo0}, we know that for $n \geq 2$, any  MDS of $\mathcal{G}_{n+1}$ is in fact the union of  MDSs in $\Upsilon_n^2$, of the four copies of $\mathcal{G}_{n}$ (i.e. $\mathcal{G}_{n}^{(1)}$, $\mathcal{G}_{n}^{(2)}$, $\mathcal{G}_{n}^{(3)}$, and $\mathcal{G}_{n}^{(4)}$) constituting $\mathcal{G}_{n+1}$. Therefore, we obtain
	\begin{align}\label{recd}
	y_{n+1}=y_{n}^4,
	\end{align}
	which, under the initial value $y_2=2$, is solved  to yield $x_n=2^{2^{2n-4}}$.
\end{proof}
Theorem \ref{NubMDSF} shows that the  number of  MDSs in graph $\mathcal{G}_n$ grows exponentially with the number of vertices $N_n$, which is similar to the number of  MISs.

\subsection{Domination number and the number of minimum dominating sets in non-fractal scale-free networks}

We finally study the domination number and the number of  MDSs in the  non-fractal scale-free network $\mathcal{G}'_n$.

Analogously to graph $\mathcal{G}'_n$,   all the  dominating sets in $\mathcal{G}'_n$ can be classified into three sets: $\Gamma^0_n$, $\Gamma^1_n$ and $\Gamma^2_n$, where $\Gamma^k_n$, $k=0,1,2$, represent the set of dominating sets, each including exactly $k$ hub vertices of $\mathcal{G}'_n$. Let $\Upsilon^k_n$, $k=0,1,2$, be the subset of $\Gamma^k_n$, where each independent set has the smallest cardinality, denoted by $\gamma^k_n$, $k=0,1,2$. Then,  the domination number $\gamma_n$ of $\mathcal{G}'_n$ can be expressed by $\gamma_n={\rm min}\{\gamma^0_n, \gamma^1_n, \gamma^2_n\}$.

\begin{theorem}
For $n \geq 3$, the domination number of non-fractal scale-free graph $\mathcal{G}'_n$ is  $\gamma_n=\frac{2^{2n-3}+4}{3}$.
\end{theorem}
\begin{proof}
Since $\gamma_n={\rm min}\{\gamma^0_n, \gamma^1_n, \gamma^2_n\}$, we first evaluate the quantities $\gamma^0_n$, $\gamma^1_n$ and $\gamma^2_n$.  In a way similar to the case of graph $\mathcal{G}_n$,  we establish the following recursive relations governing the three quantities $\gamma^0_n$, $\gamma^1_n$ and $\gamma^2_n$:
\begin{align} 
\gamma^0_{n+1}&={\rm min}\{4\gamma^0_n,2\gamma^0_n+2\gamma^1_n-1,4\gamma^1_n-2\}\,, \label{gamma0g}\\
\gamma^1_{n+1}&={\rm min}\{2\gamma^0_n+2\gamma^1_n-1,\gamma^0_n+2\gamma^1_n+\gamma^1_n-2,\notag\\
&\quad \quad\quad\quad
2\gamma^1_n+2\gamma^2_n-3\}\,, \label{gamma1g}\\
\gamma^2_{n+1}&={\rm min}\{4\gamma^1_n-2,2\gamma^1_n+2\gamma^2_n-3,4\gamma^2_n-4\}\,.\label{gamma2g}
\end{align}
Figures~\ref{dnfo0},~\ref{dnfo1}, and~\ref{dnfo2} show, respectively, all the possible  configurations of dominating sets $\Gamma^0_{n+1}$,  $\Gamma^1_{n+1}$, and $\Gamma^2_{n+1}$ for graph $\mathcal{G}'_{n+1}$. From these three figures,
we can establish Eqs.~\eqref{gamma0g},~\eqref{gamma1g}, and~\eqref{gamma2g}. By using the initial conditions $\gamma^0_3=8$, $\gamma^1_3=7$ and $\gamma^2_3=4$, Eqs.~\eqref{gamma0g},~\eqref{gamma1g}, and~\eqref{gamma2g} are solved to yield $\gamma^0_n=\frac{2^{2n-3}+3\cdot2^n-2}{3}$, $\gamma^1_n=\frac{2^{2n-3}+3\cdot2^{n-1}+1}{3}$ and $\gamma^2_n=\frac{2^{2n-3}+4}{3}$.  Hence, $\gamma_n={\rm min}\{\gamma^0_n, \gamma^1_n, \gamma^2_n\}=\gamma^2_n=\frac{2^{2n-3}+4}{3}$.
\end{proof}

\begin{figure}
\centering
    \begin{minipage}[c]{0.5\textwidth}
        \centering
        \includegraphics[width=0.27\textwidth]{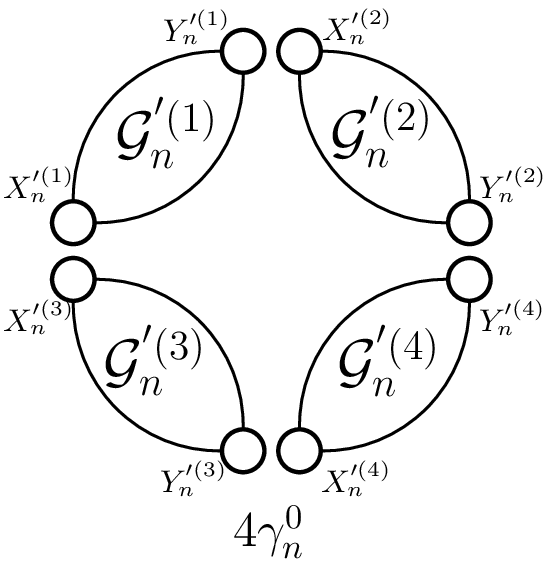}
        \includegraphics[width=0.3\textwidth]{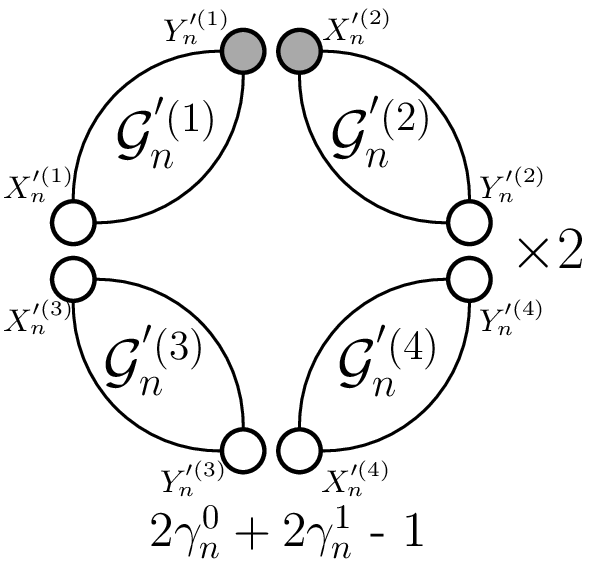}
        \includegraphics[width=0.27\textwidth]{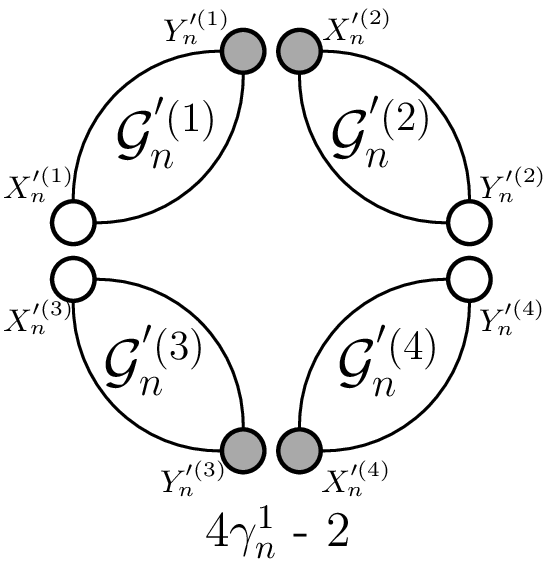}
    \end{minipage}
\caption{Illustration of all possible configurations and their sizes of dominating sets $\Gamma^0_{n+1}$ in graph $\mathcal{G}'_{n+1}$ containing $\Upsilon^0_{n+1}$.}
\label{dnfo0}
\end{figure}

\begin{figure}
\centering
    \begin{minipage}[c]{0.5\textwidth}
        \centering
        \includegraphics[width=0.3\textwidth]{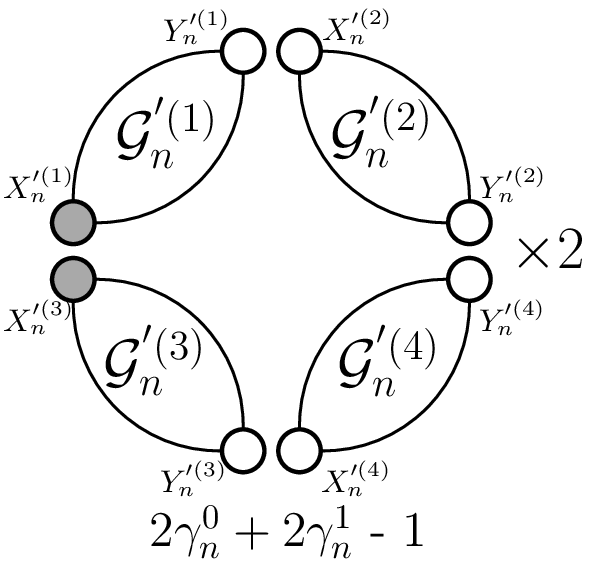}
        \includegraphics[width=0.3\textwidth]{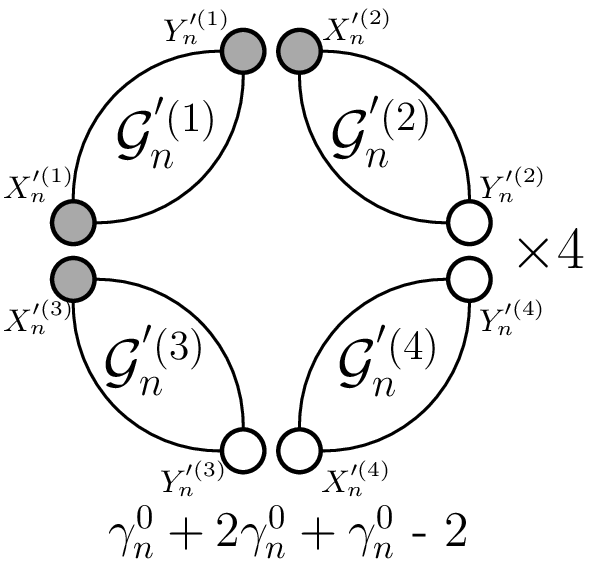}
        \includegraphics[width=0.3\textwidth]{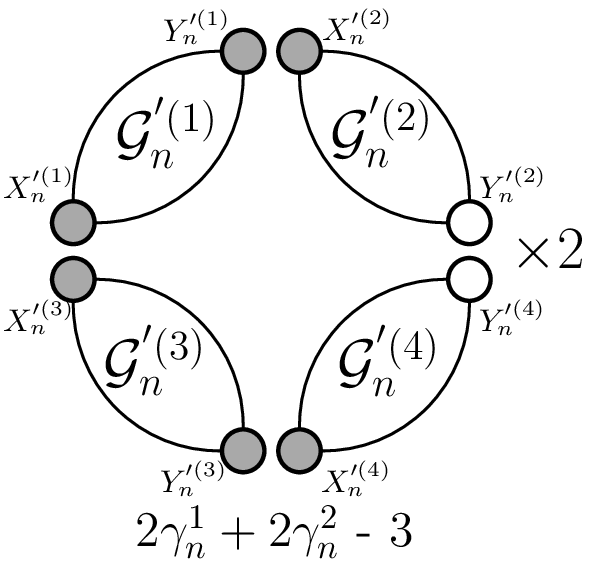}
    \end{minipage}
\caption{Illustration of all possible configurations and their sizes of dominating sets $\Gamma^1_{n+1}$ in graph $\mathcal{G}'_{n+1}$ containing $\Upsilon^1_{n+1}$.}
\label{dnfo1}
\end{figure}

\begin{figure}
\centering
    \begin{minipage}[c]{0.5\textwidth}
        \centering
        \includegraphics[width=0.27\textwidth]{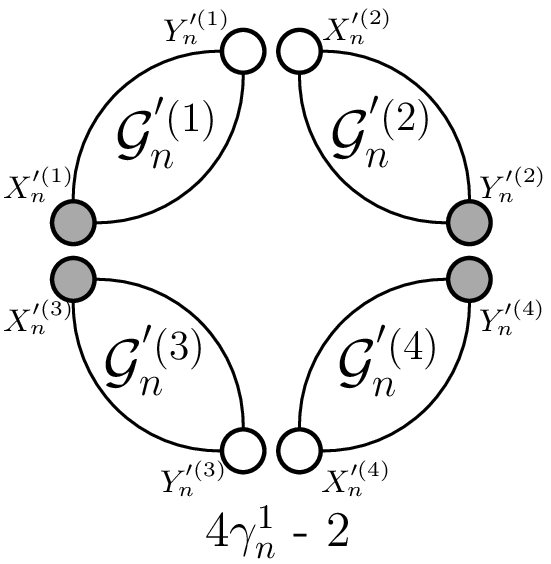}
        \includegraphics[width=0.3\textwidth]{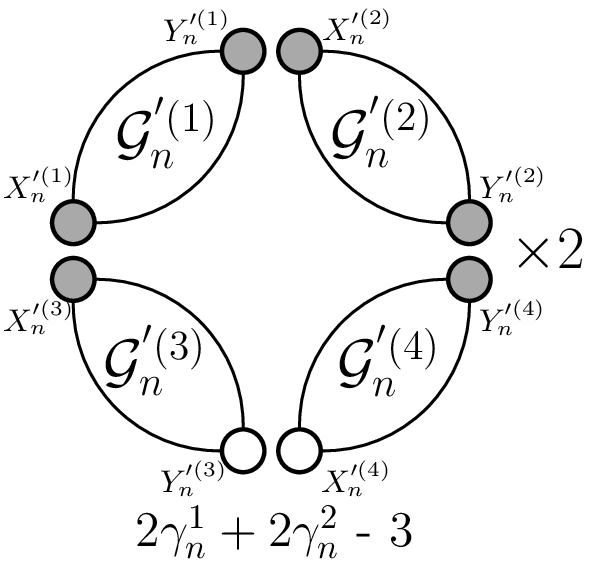}
        \includegraphics[width=0.27\textwidth]{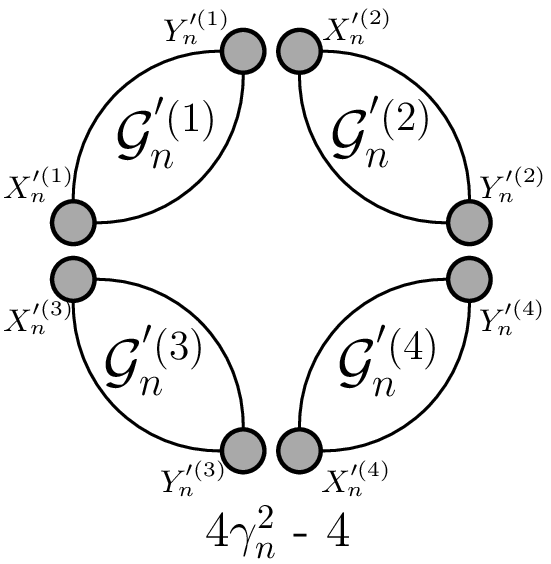}
    \end{minipage}
\caption{Illustration of all possible configurations and their sizes of dominating sets $\Gamma^2_{n+1}$ in graph $\mathcal{G}'_{n+1}$ containing $\Upsilon^2_{n+1}$.}
\label{dnfo2}
\end{figure}

\subsubsection{Number of minimum dominating sets}
In contrast to the its fractal counterpart  $\mathcal{G}_{n}$,  the non-fractal scale-free graph $\mathcal{G}'_{n}$ has only one maximum independence set for all $n \geq 3$.

\begin{theorem}\label{NumMDSg}
In  the non-fractal  scale-free graph $\mathcal{G}'_{n}$, $n \geq 3$, there exists a unique minimum  dominating set.
\end{theorem}
\begin{proof}
	 Denote by $y_n$ the number of MDSs in  $\mathcal{G}'_n$. Eq.~\eqref{gamma2g} and   Fig.~\ref{dnfo2} show that for $n \geq 3$, any  MDS of $\mathcal{G}'_{n+1}$ is actually the union of  MDSs in $\Upsilon_n^2$, of the four copies of $\mathcal{G}'_{n}$ (i.e. $\mathcal{G}_{n}^{'(1)}$, $\mathcal{G}_{n}^{'(2)}$, $\mathcal{G}_{n}^{'(3)}$, and $\mathcal{G}_{n}^{'(4)}$) forming $\mathcal{G}'_{n+1}$. Thus, one obtains $y_{n+1}=y_{n}^4$, which with the initial value $y_{2}=1$ is solved to give $y_{n}=1$ for  all $n \geq 3$.
\end{proof}
Theorem~\ref{NumMDSg} implies that  for all $n \geq 3$,  $\mathcal{G}'_n$  has a unique  MDS. Moreover,  the unique  MDSs of $\mathcal{G}'_n$, $n \geq 3$, in fact contains exactly  all the hub and border vertices in  graph $\mathcal{G}'_{n-1}$.

\section{Conclusion}

Many real-world networks simultaneously display the striking scale-free and self-similar properties. Prior works have shown that the scale-free topology has an substantial effects on the various properties of graphs, e.g., combinatorial properties.  In this paper, we studied some combinatorial problems for two self-similar scale-free networks with  identical power exponent, both of which are constructed in an iterative manner. At any iteration, the two networks have the same number of vertices and the same number of edges. Although both networks bear some resemblance, they differ in some aspects. For example, the first one is ``large-world'' and fractal, while the second one is small-world and non-fractal. By using their self-similarity and decimation technique, we provide exact expressions for the maximum number, the matching number, the independence number, and the domination number for both networks. Moreover, we find exact or recursive solutions to the number of maximum matchings, the number of MISs, and the number of MDSs for both graphs.

For the maximum matching problem, the matching number of the fractal graph is about twice that of the non-fractal graph, but in both graphs the number of maximum matchings grows exponentially with the number of total edges in the graphs. With respect to the MIS problem, the independence number of the first network is exactly half of the second network. In addition, the number of the MISs in the first graph grows exponentially with the number of vertices in the graph. In contrast, the second graph has a unique MIS. Finally, as for the MDS problem, the domination number of the fractal graph is about twice as large as its non-fractal counterpart. Moreover, the number of the MDSs in the fractal graph grows exponentially with the vertex number, while there exists a unique MDS in the non-fractal graph. Thus, although both graphs are self-similar and scale-free with the same vertex number, edge number, and power exponent, they greatly differ in the studied combinatorial aspects. Our results show that scale-free topology itself is not sufficient to characterize combinatorial properties in power-law graphs. Given the relevance of combinatorial problems to various practical scenarios, our work sheds light on better understanding the applications of combinatorial properties for scale-free networks.

\section*{Acknowledgements}

This work was supported  in part by the National Natural Science Foundation of China (Nos. 61803248, U20B2051, 61872093, and U19A2066), the National Key R \& D Program of China
(No. 2018YFB1305104 and 2019YFB2101703), and the Innovation Action Plan of Shanghai Science and Technology (Nos. 20222420800 and 20511102200).  Che Jiang was also supported by Fudan Undergraduate Research
Opportunities Program (FDUROP).


\section*{Data Availability Statement}

No new data were generated or analysed in support of this research.

\bibliographystyle{compj}
\bibliography{CombinSFBib}

\end{document}